\documentclass[aip,jmp]{revtex4-1} 

\usepackage{amsmath,amssymb,amsthm,bbm,enumerate,color,hyperref}

\newcommand{\ket}[1]{\ensuremath{\vert#1\rangle}}
\newcommand{\braket}[2]{\ensuremath{\langle #1\vert#2\rangle}}

\newcommand{\ens}[0]{\ensuremath}
\newcommand{\iE}[0]{\ens{\mathrm{i}}}

\newcommand{\Eins}[0]{\ens{\mathbbm{1}}}
\newcommand{\F}[0]{\ens{\mathbb{F}}}
\newcommand{\N}[0]{\ens{\mathbb{N}}}

\newcommand{\C}[0]{\ens{\mathbb{C}}}

\newcommand{\cB}[0]{\ens{\mathcal{B}}}

\newcommand{\cH}[0]{\ens{\mathcal{H}}}

\newcommand{\Mg}[1]{\ens{\left\lbrace #1 \right\rbrace}}
\newcommand{\MgN}[1]{\ens{\Mg{0,\dots,#1}}}
\newcommand{\MgE}[1]{\ens{\Mg{1,\dots,#1}}}

\newcommand{\ZX}{\ensuremath{Z\!X}}

\theoremstyle{plain}
\newtheorem{lem}{Lemma}[section]

\newtheorem{thm}[lem]{Theorem}

\theoremstyle{definition}
\newtheorem{defi}{Definition}[section]

\begin{document}

\title{Structure of the sets of mutually unbiased bases with cyclic
  symmetry}

\author{Ulrich Seyfarth} 
\affiliation{Max-Planck-Institut f\"ur die  Physik des Lichts, 
G\"{u}nther-Scharowsky-Stra{\ss}e 1, Bau 24,
91058 Erlangen, Germany}

\author{Luis L. S\'{a}nchez-Soto} 
\affiliation{Max-Planck-Institut f\"ur die Physik des Lichts, 
 G\"{u}nther-Scharowsky-Stra{\ss}e 1, Bau 24, 
 91058 Erlangen, Germany} 
 \affiliation{Department f\"{u}r Physik, 
   Universit\"{a}t Erlangen-N\"{u}rnberg, Staudtstra{\ss}e 7,
  Bau 2, 91058 Erlangen, Germany} 
  \affiliation{Departamento de \'Optica, Facultad de F\'{\i}sica, 
  Universidad Complutense, 28040~Madrid, Spain}

\author{Gerd Leuchs} 
\affiliation{Max-Planck-Institut f\"ur die Physik des Lichts, 
 G\"{u}nther-Scharowsky-Stra{\ss}e 1, Bau 24, 
 91058 Erlangen, Germany} 
 \affiliation{Department f\"{u}r Physik, 
   Universit\"{a}t Erlangen-N\"{u}rnberg, Staudtstra{\ss}e 7,
  Bau 2, 91058 Erlangen, Germany}

\date{\today}



\begin{abstract}
  Mutually unbiased bases that can be cyclically generated by a single
  unitary operator are of special interest, since they can be readily
  implemented in practice. We show that, for a system of qubits,
  finding such a generator can be cast as the problem of finding a
  symmetric matrix over the field $\F_{2}$ equipped with an
  irreducible characteristic polynomial of a given Fibonacci
  index. The entanglement structure of the resulting complete sets is
  determined by two additive matrices of the same size.
\end{abstract}

\maketitle
\section{Introduction}
\label{sec:intro}
Complementarity distinguishes the world of quantum phenomena from the
realm of classical physics~\cite{Bohr:1928fk}. At a fundamental level,
mutually unbiased bases (MUBs) provide, perhaps, the most accurate
statement of complementarity. This notion emerged in the seminal work
of Schwinger~\cite{Schwinger:1960a,Schwinger:1960b,Schwinger:1960c}
and has gradually turned into a cornerstone of quantum information
(see Ref.~\onlinecite{Durt:2010hc} for a comprehensive review).  MUBs
have long been known to provide an optimal scheme for quantum
tomography~\cite{Wootters:1986xy,Wootters:1989vn} and are central to
the formulation of the discrete Wigner
function~\cite{Wootters:1987qf,Asplund:2001,Gibbons:2004bh,
  Vourdas:2004nr,Bjork:2008fv}. They have also been used in
cryptographic protocols~\cite{Bechmann:2000ly,Bourennane:2001qd,
  Nikolopoulos:2005yg,Yu:2008wd, Xiong:2012fe,Mafu:2013qc}, in quantum
error correction codes~\cite{Gottesman:1996ve, Calderbank:1996yo,
  Calderbank:1997qf, Griffiths:2005rf,Pawowski:2010if}, and in quantum
game theory, in particular to provide a solution to the mean king
problem~\cite{Aharonov:2001dq,Englert:2001cr,Aravind:2003nx,
Klappenecker:2005lf,Hayashi:2005oq,Paz:2005dz,Durt:2006fu,
Kimura:2006kl,Kalev:2013bx,Revzen:2013ax}.

For a $d$-dimensional quantum system, it has been shown that the
number of MUBs is at most $d+1$~\cite{Ivanovic:1981tg}. Actually, such
a complete set of MUBs exists whenever $d$ is prime or power of
prime~\cite{Calderbank:1997ao}.  Remarkably though, there is no known
answer for any other values of $d$, although there are some attempts
to find a solution to this problem in some simple cases, such as
$d=6$~\cite{Grassl:2004iz,Butterley:2007fe,Brierley:2008ja,
  Brierley:2009mw,Raynal:2011xc,McNulty:2012ez} or when $d$ is a
nonprime integer squared~\cite{Archer:2005bz,Wocjian:2005pw}.  Recent
work suggests that the answer to this question may well be related
with the non-existence of finite projective planes of certain
orders~\cite{Saniga:2004xa,Bengtsson:2004ol,Weigert:2010ev} or with
the problem of mutually orthogonal Latin squares in
combinatorics~\cite{Zauner:2011fv,Wootters:2006vk,
  Paterek:2009zs,Hall:2010ir}. Furthermore, MUBs have interesting
connections to symmetric informationally complete
positive-operator-valued measures~\cite{Appleby:2009vn} and complex
$t$-designs~\cite{Klappenecker:2005qd,Gross:2007cs}.

Many explicit constructions of MUBs in prime power dimensions have
been proposed~\cite{Bandyopadhyay:2002lm,Parthasarathy:2004pf,
Pittenger:2004vs,Durt:2005yw,Klimov:2005qb,Kibler:2006kn}. However,
irrespective of the approach, one has to face an intriguing question:
different complete sets of MUBs exist with distinct entanglement
properties~\cite{Lawrence:2002ij,Lawrence:2004xj,Romero:2005dz,
Garcia:2010hr,Wiesniak:2011kl,Rehacek:2013jt,Spengler:2013pi}.  For
the experimentalist, this information is of utmost importance, because
the complexity of implementing a given set greatly depends on how many
registers need to be entangled. Note carefully that  this entanglement
structure is different from the inequivalence of different sets of
MUBs~\cite{Godsil:2009,Kantor:2012}.
 
In even prime-power dimensions, complete sets of MUBs can be
shaped in a cyclic manner, as the multiples of a single
generating basis~\cite{Chau:2005wj,Gow:2007dz,Kern:2010lc,
  Seyfarth:2011ru,Seyfarth:2012sl}. This procedure rests on the
properties of the so-called Fibonacci
polynomials~\cite{Hoggart:1973mq} and leads directly to quantum
circuits that can be used for a simple practical realization of these
bases.

In this work, we present a method to setup sets of cyclic MUBs with
different entanglement structures. This is accomplished by unveiling
certain structures within the different sets, which are related either
to a field, an additive group or an additive semigroup. The key idea
is a two-step generalization that softens the properties of the
starting field structure, while preserving the proper features
inherited from the Fibonacci polynomials.

This paper is organized as follows: In Section~\ref{sec:prelim} we
introduce the basic tools and definitions needed for the rest of our
expos\'e. The main results are covered in Section~\ref{sec:method},
which starts with complete sets possessing a field structure and
discusses their extension to sets with a group and a semigroup
structure. The generality of the method is confirmed in
Section~\ref{sec:cmp}, whereas our conclusions are summarized in
Section~\ref{sec:conclusion}.

\section{Preliminaries}
\label{sec:prelim}

\subsection{Mutually unbiased bases}

\begin{defi}[Mutually unbiased bases]
  \label{defi:prelim:mubs} \hfill \\
  Let $\cH= \C^d$ be a $d$-dimensional complex Hilbert space. A pair
  of orthonormal bases $\mathcal{B}_{j} = \{ | \psi_{\ell}^{(j)}
  \rangle \}$ and $\mathcal{B}_{j^{\prime}} = \{ |
  \psi_{\ell^{\prime}}^{(j^{\prime})} \rangle \}$ (with $j \neq
  j^{\prime})$ is said \emph{mutually unbiased} if
  \begin{align}
    | \braket{\psi_{\ell}^{(j)}}{\psi_{\ell^{\prime}}^{(j^{\prime})}}
    |^{2} = \frac{1}{d}
  \end{align}
  for all $\ell , \ell^{\prime} \in \MgE{d}$.
\end{defi}

In physical terms, this means that if the system is prepared in a
state of the first basis, then all outcomes are equally probable when
we conduct a measurement that probes the states of the second basis.
Familiar examples are the spin states of a spin-1/2 particle for two
perpendicular directions or any basis and its Fourier transform for
any dimension $d$.

When the bases in the orthonormal set $\mathfrak{S} = \Mg{ \cB_{0},
  \ldots, \cB_{r}}$, with $r \in \N_{0}$, are pairwise unbiased, we say that
$\mathfrak{S} $ is a set of MUBs. When such a set contains the maximal
number of elements, $d+1$, it is called a \emph{complete} set of MUBs.

As heralded in the introduction, for even prime-power dimensions,
complete sets of MUBs can be constructed in a cyclic way. To
understand this point, we observe that any basis in $\C^{d}$ can be
always identified with a unitary matrix $U \in M_{d} (\C)$ ($M_{d}
(\C)$ stands for the $d \times d$ matrices over $\C$) acting on this
space, since the columns of $U$ define an orthonormal basis and vice
versa. Since $U^{2}, U^{3}, \ldots$ are unitary, they also define
bases.

\begin{defi}[Cyclic mutually unbiased bases]
  \label{defi:prelim:cmubs} \hfill \\
  A complete set $\mathfrak{S} = \Mg{\cB_{0}, \ldots , \cB_{d}}$ of
  MUBs is called \emph{cyclic}, if there exists a unitary matrix $U
  \in M_{d} (\C)$, such that the columns of $U, U^{2}, U^{3}, \ldots,
  U^{d+1} = \openone_{d}$ coincide with the bases in $ \mathfrak{S}$.
\end{defi}

To classify different sets of MUBs, the notion of \emph{equivalence}
has to be established. 

\begin{defi}[Equivalence of complete set of mutually unbiased bases]
  \label{defi:prelim:equivalence} \hfill\\
  Let $\mathfrak{S} = \Mg{B_{0}, \ldots, B_{d}}$ and
  $\mathfrak{S}^{\prime} = \Mg{B^{\prime}_{0}, \ldots,
    B^{\prime}_{d}}$ be two complete sets of MUBs. Both sets are said
  to be equivalent if there holds
  \begin{align}
    \cB^{\prime}_j = U \cB_{\pi(j)} W_j,
  \end{align}
  for any unitary matrix $U \in M_{d} (\C)$, a permutation $\pi$ of
  $\MgN{d}$ and monomial matrices $W_j$ with $j \in \MgN{d}$.
\end{defi}

We recall that a matrix $W$ is called monomial, if it can be expressed
as $ W = D \Pi$, with $D = \mathrm{diag}(\lambda_{1}, \ldots ,
\lambda_{d})$ is diagonal and $\Pi$ is a permutation. We assume
$|\lambda_{i}| = 1$, with $i \in \MgE{d}$; i.e., $W$ is unitary.

\subsection{Pauli operators and entanglement structure}

Let us consider a prime-power dimension $d=p^m$, with $p$ prime and $m
\in \N$. On the Hilbert space $\cH = \C^{p}$, with canonical
orthonormal basis $\{ |i \rangle : i \in \F_{p}\}$ we define the Pauli
operators by
\begin{align}\label{eqn:prelim:ZX}
  Z \ket{i} = \omega^i \ket{i} \, , \qquad \qquad X \ket{i} = \ket{i
    \oplus 1},
\end{align}
where $\omega = \exp{(2 \pi \iE /d)}$ is the first $p$-th root of
unity and the addition $\oplus$ must be understood $\bmod \ p$. 

This concept can be generalized to the Hilbert space $\cH = \C^{d}$ by
introducing the $2m$-dimensional vector 
$\vec{a} = (a^z_1,\ldots,a^z_m; a^x_1,\ldots,a^x_m)^t \in \F_p^{2m}$,
and in terms of it, the set
\begin{align}
\label{eqn:bandy:ZXa}
 \ZX (\vec a) = 
\begin{cases}
     (-\iE)^{a^z_1 a^x_1} Z^{a^z_1} X^{a^x_1} \otimes\dots\otimes
     (-\iE)^{a^z_m a^x_m} Z^{a^z_m} X^{a^x_m}, &  p=2,\\
     Z^{a^z_1} X^{a^x_1} \otimes\dots\otimes Z^{a^z_m} X^{a^x_m}, &
    p \neq 2.
\end{cases}
\end{align}
The set of these Pauli operators is just the Weyl-Heisenberg group
factorized by its center.  Two Pauli operators defined by the vectors
$\vec a, \vec b \in \F_p^{2m}$ obey the symplectic commutation
relation
\begin{align}
 \ZX(\vec a) \cdot \ZX(\vec b) = 
 \omega^{(\vec a, \vec b)}  \ZX(\vec b) \cdot \ZX(\vec a),
\end{align}
with the symplectic inner product defined as
\begin{align}
 (\vec a, \vec b) = \sum\limits_{k=1}^m 
 a_k^z b_k^x -  a_k^x b_k^z \, .
\end{align}

According with the ideas in Ref.~\onlinecite{Bandyopadhyay:2002lm},
complete sets of MUBs arise straightforwardly from a
partition of the set of Pauli operators into $d+1$ subsets of $d-1$
commuting operators, called \emph{classes}. These classes
$\mathfrak{C}^{\prime}_j$ are given for $j \in \MgN{d}$ by
\begin{align}
  \label{eqn:prelim:classes}
  \mathfrak{C}^{\prime}_j = \mathfrak{C}_j \cup \{ \Eins_d \} = \{
  \ZX(\vec{a}) : \vec{a} = G_j  \vec{c}: \vec{c} \in \F_p^m \},
\end{align}
in terms of the generators $G_j$.

The Pauli operators within a class $\mathfrak{C}_j$ commute by
construction. But, we can check if the Pauli operators corresponding
to different subsystems also commute. If so, the property measured by
the corresponding basis will be a property for which this subsystem is
not entangled with the rest of the system. In principle, all possible
partitions of the number of subsystems $m \in \N$ are possible, from
\emph{completely factorizable} to \emph{fully entangled}
systems. Finally, any set of basis can be classified by a vector $\vec
n$, where each element counts the number of bases having a specific
entanglement structure. The length of $\vec n$ is given by the number
of partitions of $m$, and the first entry gives the number of
completely factorizable bases.

\subsection{Fibonacci polynomials}

Our analysis of complete sets of cyclic MUBs will rely on the
properties of \emph{Fibonacci polynomials}, which are a generalization
of the well-known \emph{Fibonacci sequence}.
\begin{defi}[Fibonacci polynomials]
  \label{defi:prelim:mubsFib} \hfill \\
  Over an arbitrary field $K$, we define the Fibonacci polynomials
  $F_{n} (x)$ ($n \in \N_{0}$) by the recursion relation
  \begin{align}
    F_{j+1}(x) = xF_{j} (x) + F_{j-1} (x),
  \end{align}
  with $F_0 = 0$ and $F_1 = 1$.
\end{defi}

In our context, we exclusively deal with the ground field $\F_{2}$ and
possibly its extensions. As the Fibonacci sequence, also the Fibonacci
polynomials can be constructed using a generator, namely,
\begin{align}
  \mathcal{A} =
  \begin{pmatrix}
    x & 1\\
    1 & 0
  \end{pmatrix}
  \equiv
  \begin{pmatrix}
    F_2(x) & F_1(x)\\
    F_1(x) & F_0(x)
  \end{pmatrix},
\end{align}
so that the powers of this generator are
\begin{align}\label{eqn:prelim:Aj}
  \mathcal{A}^j=
  \begin{pmatrix}F_{j+1}(x) & F_j(x)\\ F_j(x) &
    F_{j-1}(x)\end{pmatrix}. 
\end{align}
One can verify that
\begin{align}
\label{eq:lusa}
  F_j(x) F_{k+1}(x) + F_{j-1}(x) F_k(x) = F_{k+1}(x).
\end{align}

The Fibonacci numbers satisfy several well-known divisibility
relations; we will need their counterparts for Fibonacci polynomials
over $K = \F_{2}$.  

\begin{defi}[Fibonacci index]
\label{eqn:prelim:fibidx}\hfill\\
  The Fibonacci index of an irreducible polynomial $p \in K[x]$ is defined
  as the minimum number $n \in \N$, such that $p(x)$ divides $F_n(x)$.
\end{defi}

If the Fibonacci polynomials are defined over $\F_2$, the Fibonacci
index of any irreducible polynomial is either a divisor of $2^m-1$ or
$2^m+1$, with $m$ being the degree of $p(x)$~\cite{Goldwasser:2002pz}.

\section{Method}
\label{sec:method}

The aim of this article is to investigate a method to generate
complete sets of cyclic MUBs with different entanglement
structures. According to intrinsic structures, three different
constructions will be discussed: we start with a basic method
and present then two consecutive generalizations.

\subsection{Field-based sets}

In Ref.~\onlinecite{Seyfarth:2012sl} it has been shown how to
construct complete sets of cyclic MUBs in dimension $d=2^m$, $m \in
\N$.  This problem can accordingly be reduced to finding a suitable
symplectic stabilizer matrix $C\in M_{2m}
(\F_{2})$~\cite{Kern:2010lc}. In fact, such a matrix can be written as
\begin{align}
\label{eqn:field:C}
  C = \begin{pmatrix}
    B & \Eins_m \\
    \Eins_m & 0_m
  \end{pmatrix},
\end{align}
where $B$ is a symmetric and invertible matrix whose characteristic
polynomial has Fibonacci index $d+1$. Each solution leads to a
complete set of cyclic MUBs.

It is worth noticing that the powers of $C$ can be easily written,
according to Eq.~\eqref{eqn:prelim:Aj}, as
\begin{align}\label{eqn:field:Cj}
  C^j =
  \begin{pmatrix}
    F_{j+1}(B) & F_{j}(B) \\
    F_{j}(B) & F_{j-1}(B)
  \end{pmatrix}.
\end{align}

The generators of the cyclic sets can be written as
\begin{align}
  G_j = C^j G_0 \,,  
 \qquad \qquad G_0 =
  \begin{pmatrix}
    \Eins_m\\ 0_m
  \end{pmatrix},
\end{align}
so, using \eqref{eqn:field:Cj}, we have
\begin{align}
  G_j = \begin{pmatrix}
    F_{j+1} (B)\\
    F_{j} (B)
  \end{pmatrix}.
\end{align}

For our purposes in what follows, it will prove convenient to rewrite
these generators in a \emph{standard form} as
\begin{align}
  \label{eqn:field:barGj}
  \bar{G}_j =
  \begin{pmatrix}
    F_{j+1} (B) F_{j}^{-1} (B)  \\
    \Eins_m
  \end{pmatrix},
\end{align}
for $j \in \MgE{d}$ and $\bar G_0 \equiv G_0$. This transformation is
possible by exploiting the properties of the generation of the classes
in \eqref{eqn:prelim:classes}, where each invertible square matrix
multiplied from the left to all elements $\vec c$ is only permuting
these elements, thus multiplying any generator matrix $G_j$ from the
right by an invertible square matrix of appropriate size does not
change the generated set.

\begin{lem}[Generators are finite field representation]
  \label{lem:field:represent}\hfill\\
  The set of $m \times m$ upper submatrices of the generators $\{
  G_{j} \}$  of a
  complete set of cyclic MUBs in the standard form
  \eqref{eqn:field:barGj} is a representation of the finite field
  $\F_{2^m}$.
\end{lem}
\begin{proof}
  For a complete set of cyclic MUBs, all the generators $\bar G_j$
  have to be distinct for $j \in \MgE{d}$. The matrices $F_{j+1} (B)
  F^{-1}_{j} (B)=p_j(B)$ are polynomials that can be taken modulo the
  characteristic polynomial of $B$. This is irreducible and has order
  $m$. As there are $d=2^m$ different generators, each element of the
  field can be represented, with $p_j(B)=0_m$ being the neutral
  element of addition and $p_{d/2}(B)=\Eins_m$ the neutral element of
  multiplication.
\end{proof}

Lemma \ref{lem:field:represent} justifies that we call these sets
\emph{field-based sets}. All sets exhibiting this structure, will have
three bases (and therefore three corresponding generators) with an
entanglement structure that is completely factorizable, namely
\begin{align}
  G_0 =
  \begin{pmatrix}
    \Eins_m \\ 0_m
  \end{pmatrix}, \qquad \bar G_{d/2} =
  \begin{pmatrix}
    \Eins_m \\ \Eins_m
  \end{pmatrix}, \qquad G_d =
  \begin{pmatrix}
    0_m \\ \Eins_m
  \end{pmatrix},
\end{align}
where the classes make the sets of all Pauli $Z$, $Y$, and $X$
operators, respectively. Therefore, to change the number of completely
factorizable bases within a set, we need to adapt the form of the
stabilizer matrix $C$ in \eqref{eqn:field:C}.  In the following, we
will show how to build complete sets of MUBs with only two and one
completely factorizable bases by two generalizations. For this, we
start from an approach for the generators in standard form and try to
calculate the corresponding stabilizer matrix.

\subsection{Group-based sets}

To reduce the number of completely factorizable bases, we break down
the field structure discussed in the last section to an additive group
structure. This excludes the neutral element of the multiplication, as
it accounts for one of the completely factorizable bases.

If, in the standard form,  the new generators are
\begin{align}
\label{eqn:group:barGj}
  \bar G_j = 
\begin{pmatrix}
F_{j+1} (B) F_{j}^{-1} (B) R \\
    \Eins_m 
\end{pmatrix} ,
\qquad 
G_0  = 
\begin{pmatrix}
\Eins_m \\ 0_m 
\end{pmatrix} ,
\end{align}
Theorem~$4.4$ of~\cite{Bandyopadhyay:2002lm} would still be
fulfilled, as long as $R \in M_{m}(\F_2)$ is invertible. This  leads to a
stabilizer matrix
\begin{align}
\label{eqn:group:C}
  C = \begin{pmatrix}
B & R\\ R^{-1} & 0_m
\end{pmatrix},
\end{align}
or, in terms of Fibonacci polynomials,
\begin{align}
  C^j =
  \begin{pmatrix}
    F_{j+1}(B) & F_j(B) R\\
    R^{-1} F_j(B) & R^{-1} F_{j-1}(B) R
  \end{pmatrix}.
\end{align}
By setting $R=\Eins_m$, it is obvious that this stabilizer is a
generalization of the one used for the field-based sets. From
Theorem~$4.4$ of Ref.~\onlinecite{Bandyopadhyay:2002lm} all matrices
$p_j(B) R$ have to be symmetric. Furthermore, an analogous argument as
for the field-based sets leads to the condition, that the
characteristic polynomial of $B$ should have Fibonacci index $d+1$. By
the following lemma we can substantially reduce one of the conditions.

\begin{lem}[Symmetrizer applies to polynomials]
  \label{lem:group:symmetry}\hfill\\
  Let $K$ denote an arbitrary field. Let us assume given the
  invertible matrix $B \in M_{m} (K)$, there exists an invertible and
  symmetric \emph{symmetrizer} $R \in M_{m} (K)$ such that $BR$ is
  symmetric. Then, all polynomials of the form $p(B) R$, with $p(B)
  \in K[x]$, are also symmetric.
\end{lem}
\begin{proof}
  If $R$ and $BR$ are symmetric, then for any matrix $B^k R$ with $k
  \in \N$ we have
  \begin{align}
    (B^k R)^t = R^t B^t (B^{k-1})^t = B R^t B^t (B^{k-2})^t = \ldots =
    B^k R.
  \end{align}
  As sums of symmetric matrices are symmetric, this holds for all
  polynomials $p(B)$.  Considering the cases $k=0$ and $k=1$, the
  converse statement is obvious.
\end{proof}
This Lemma and the previous discussion lead to the following final set
of conditions to be fulfilled to construct a complete set of cyclic
MUBs with three or two completely factorizable bases for
invertible $R,B \in M_m(\F_2)$:
\begin{enumerate}[(i)]
\item $R$ and $B R$ are symmetric.
\item The characteristic polynomial of $B$ has Fibonacci index $d+1$.
\end{enumerate}
For each valid $B$ a corresponding \emph{symmetrizer} matrix $R$
exists which is additionally a square root of
unity~\cite{Gow:1980hx}. An algorithm can be found in
Ref.~\onlinecite{Djokovic:2003qt,Venkaiah:1988ph}.

To construct only sets which have not more than two completely
factorizable bases, $R$ has to destroy the field structure as will be
seen by the following lemma.

\begin{lem}[Destruction of field structure]
\label{lem:group:only2}\hfill\\
  Iff the symmetric matrix $R \in M_{m} (\F_2)$ does not
  equal any polynomial of $B \in M_{m} (\F_2)$, where $R$ and $B$ are
  chosen according to the conditions of group-based sets, the
  resulting complete set of cyclic MUBs has exactly two completely
  factorizable bases.
\end{lem}
\begin{proof}
  If $R$ is chosen to be a polynomial of $B$, say $q(B)$, the
  generators $\bar G_j$ in  \eqref{eqn:group:barGj} will read 
  \begin{align}
    \bar G_j = 
\begin{pmatrix}
p_j(B) q(B) \\
      \Eins_m 
\end{pmatrix}, 
\qquad \
G_0 = 
\begin{pmatrix}
\Eins_m \\ 0_m 
\end{pmatrix},
  \end{align}
  and thus the classes will be a permuted version of the field-based
  sets given in \eqref{eqn:field:barGj}. Conversely, as soon as $R$
  does not equal any polynomial of $B$, none of the generators will
  equal $(\Eins_m, \Eins_m)^t$, thus defining the set of Pauli $Y$
  operators. But any generator in the form $(p(B)R, \Eins_m)^t$ will
  have off-diagonal entries in $p(B)R$ for $R \neq \Eins_m\equiv B^0$
  and $p(B) \neq 0_m$ being a non-zero polynomial of $B$. Therefore,
  none of the generators can produce a completely factorizable basis.
\end{proof}

It is obvious, that all matrices $R^{\prime}$ lead to a set with the same
entanglement properties as a set with $R$, as long as $R^{\prime}= q(B) R$
holds for any non-zero polynomial $q$ of $B$, as it permutes only the set of bases.

\subsection{Semigroup-based sets}

The second  generalization is to further restrict the additive
group to an additive semigroup structure, which additionally excludes
from the set the neutral element of addition. This is achieved by
adding a symmetric matrix $A \in M_{m} (\F_2)$ to the matrices
$F_{j+1} (B) (F_{j} (B))^{-1} R$ with $j \in \MgE{d}$ in order to
preserve Theorem~$4.4$ in Ref.~\onlinecite{Bandyopadhyay:2002lm}.  We
will call these sets \emph{semigroup-based sets}.

In standard form, the generators  look like
\begin{align}
\label{eqn:semigroup:barGj}
  \bar G_j = 
\begin{pmatrix} 
F_{j+1}(B) (F_j (B))^{-1} R + A \\
    \Eins_m
\end{pmatrix} ,
\qquad 
G_0 = 
\begin{pmatrix}
\Eins_m \\ 0_m 
\end{pmatrix}, 
\end{align}
To find the corresponding stabilizer matrix $C$, we note that these generators correspond to those 
given by the product of powers of $C$ with $G_0$ and 
\begin{align}
  G_j = 
 \begin{pmatrix} 
F_{j+1}(B) + A R^{-1} F_j (B)\\
  R^{-1} F_j   (B)
\end{pmatrix}.
\end{align}
As $C G_0 = G_1$ should hold, the first column of $C$ has to be equal to
$G_1$. So we can approach $C$ as $C(x,y)$ with the two free parameters
$x,y \in M_m(\F_2)$ and fix them by applying $C G_1 = G_2$ as
\begin{align}
  \begin{pmatrix}
    B + A R^{-1} & x\\ R^{-1} & y
  \end{pmatrix}
  \begin{pmatrix}
    B + A R^{-1}\\ R^{-1}
  \end{pmatrix}
  =
  \begin{pmatrix}
    B^2 + B A R^{-1} + A R^{-1} B + AR^{-1} AR^{-1}  + x R^{-1}\\
    R^{-1} B + R^{-1} A R^{-1} + y R^{-1}
  \end{pmatrix},
\end{align}
so that
\begin{align}
  G_2 =
  \begin{pmatrix}
    B^2 + 1 + A R^{-1} B\\
    R^{-1} B
  \end{pmatrix}.
\end{align}
Solving these equations we find a new form of the stabilizer matrix,
namely
\begin{align}\label{eqn:semigroup:C}
  C=
  \begin{pmatrix}
    B + A R^{-1} & R + BA + A R^{-1} A\\
    R^{-1} & R^{-1} A
  \end{pmatrix}.
\end{align}
Setting $A=0$ leads to \eqref{eqn:group:C}, confirming that we have
found a further generalization of the group-based construction.  As a
next step, it is helpful to calculate the powers of $C$, as we can
easily check further properties of this matrix.  We find once again a
very compact form, in terms of the Fibonacci polynomials;
\begin{align}
\label{eqn:semigroup:Cj}
  C^j =
  \begin{pmatrix}
    F_{j+1}(B) + AR^{-1} F_j(B)  &  F_{j+1}(B) A + F_j(B) R + AR^{-1} [F_j(B) A + F_{j-1}(B)  R]\\
    R^{-1} F_j(B) & R^{-1} [F_j(B) A + F_{j-1}(B) R]
  \end{pmatrix}.
\end{align}
Equation~\eqref{eqn:semigroup:Cj} allows us to check that
$C^{d+1}=\Eins_{2m}$.  As long as the Fibonacci index of $B$ is $d+1$
as in the former constructions,  $C^{j} \neq \Eins_{2m}$ for any $j \in
\MgE{d}$ if one  considers only the term $R^{-1} F_j(B)$. Finally, we can
check if all powers of $C$ permute the generators as expected, namely
by checking if $C^j G^k = G^{(j+k)\;\mathrm{mod}\;(d+1)}$:
\begin{align}
  C^j G^k =
  \begin{pmatrix}
    F_{j+1}(B) F_{k+1}(B) + F_{j}(B) F_{k}(B) + A R^{-1} [  F_{j}(B)
    F_{k+1}(B) + F_{j-1}(B) F_{k}(B)  ] \\
    R^{-1} [ F_j(B) F_{k+1}(B) + F_{j-1}(B) F_k(B) ]
  \end{pmatrix}.
\end{align}
Using the general recursion relation \eqref{eq:lusa} we  find
\begin{align}
  C^j G^k =
  \begin{pmatrix}
    F_{k+j+1}(B) A R^{-1} F_{k+j}(B) \\
    R^{-1} F_{k+j}(B)
  \end{pmatrix},
\end{align}
which is the expected result. 

The conditions to be fulfilled to construct a complete set of cyclic
MUBs, using the stabilizer \eqref{eqn:semigroup:C} for
$A \in M_m( \F_2)$ and for invertible $R,B \in M_m(\F_2)$ are:
\begin{enumerate}[(i)]
\item $A$, $R$, and $B R$ are symmetric.
\item The characteristic polynomial of $B$ has Fibonacci index $d+1$.
\end{enumerate}
This generates complete sets of cyclic MUBs with three,
two and one completely factorizable bases. To have only
sets with one factorizable basis, we have to choose $R$ according to
Lemma~\ref{lem:group:only2} and add a similar condition which breaks
down the additive group structure into an additive semigroup
structure. Regarding \eqref{eqn:semigroup:barGj}, we have to
guarantee that no generator of the form $(\Eins_m, \Eins_m)^t$ or
$(0_m, \Eins_m)^t$ is produced. This additional condition is given by
the following lemma.

\begin{lem}[Destruction of additive group structure]
\label{lem:semigroup:only1}\hfill\\
  Iff the symmetric matrix $A \in M_m(\F_2)$ does not equal
  any of the matrices $p_j(B)R + D_l$, where $R$ and $B$ are
  chosen according to the conditions of semigroup-based sets, with
  $p_j(B)$ being the polynomials in $B$ and $j \in \MgE{d}$, the
  resulting complete set of cyclic MUBs has exactly one completely
  factorizable basis. The matrix $D_l$ denotes the diagonal matrix for
  which $(d_{ii}) = l_i$ holds for $l = (l_1,\ldots, l_m)$, $i \in
  \MgE{m}$ and with $l_i \in {0,1}$.
\end{lem}
\begin{proof}
  If $A$ is chosen to be a polynomial of $B$, say $q(B)$, multiplied
  by $R$, the class for which $q(B) = p_j(B)$ holds will produce the
  set of Pauli $X$ operators, therefore this possibility has to be
  excluded. Furthermore, if instead of $0_m$, a matrix $D_l$ with only
  diagonal entries appears in the $m \times m$ upper submatrix of the
  generator which would then equal $(D_l, \Eins_m)^t$, for all cases
  of $D_l$ a completely factorizable basis would be created. As long
  as this is not the case, the off-diagonal parts of the $m \times m$
  upper submatrix of the generators in standard form will lead to
  bases which are not completely factorizable.
\end{proof}

\section{Completeness}
\label{sec:cmp}

Finally, we have to guarantee that our method produces MUBs with all
the possible entanglement structures within the same \emph{equivalence
  class} (in the sense of Definition \ref{defi:prelim:equivalence})
and which are allowed by the scheme provided by Bandyopadhyay \emph{et
al}~\cite{Bandyopadhyay:2002lm}. Of course, this does not include
sets which have no completely factorizable bases.

The definition of equivalence can be restated in terms of the class
generators $G_j$ with $j \in \MgN{d}$, an arbitrary symplectic matrix
$f \in \mathrm{Sp}_{2m} (\F_2 )$ and a permutation $Q$ of the elements
within a class and a permutation $\pi_{j}$ of the index set of the set
of bases.~\cite{Seyfarth:2012sl} By \eqref{eqn:prelim:classes}, the
permutation $Q$ is of no relevance. As long, as we consider
only the set of generators $G_j$ and not their ordering, also the
permutation $\pi_{j}$ becomes irrelevant. So, we are left with the
symplectic matrix $f$ that, in general, would be given by
\begin{align}
  f = \begin{pmatrix}s&t\\u&v\end{pmatrix},
\end{align}
where $s,t,u,v \in M_m(\F_2)$. We start with a lemma:

\begin{lem}[Equivalence of different sets]
\label{lem:cmp:equivalence}\hfill\\
  The semigroup-based sets with the stabilizer matrix given by
\eqref{eqn:semigroup:C} are equivalent to the field-based
  sets with the stabilizer matrix given by 
  \eqref{eqn:field:C}.
\end{lem}
\begin{proof}
  For a certain choice of the symplectic matrix  $f$ we should be able to
  convert field-based sets into semigroup-based sets. This can be
  accomplished by multiplying the generators  by
  \begin{align}
\label{eq:fuf}
    f = \begin{pmatrix}
   s & t\\
   0_m& v
  \end{pmatrix},
  \end{align}
  where $s,t,v \in M_m (\F_2)$ and  gives
  \begin{align}
    f G_0 = 
 \begin{pmatrix} 
s\\ 0_m
 \end{pmatrix},
\qquad 
 f \bar G_j
    = \begin{pmatrix} 
  s p_j(B) s^{-1} s^t s + t s^t\\
      \Eins_m
\end{pmatrix},
  \end{align}
where we can recognize $s p_j(B) s^{-1}$ as
  $p_j(B')$, $s^t s=R$ and $ t s^t = A$ if $s$ and $t$ are chosen
  accordingly. If we set $t=0_m$, also the group-based sets with the
  stabilizer matrix given by \eqref{eqn:group:C} belong to
  the same equivalence class.
\end{proof}

Finally, we confirm  the completeness of our scheme.

\begin{thm}[Completeness of the construction scheme]
\label{thm:cmp:completeness}\hfill\\
  The stabilizer matrix \eqref{eqn:semigroup:C} leads to complete sets
  of MUBs with all the possible entanglement structures. Moreover,
  they are equivalent to the sets generated via  the stabilizer matrix
  \eqref{eqn:field:C}. 
\end{thm}
\begin{proof}
  By Lemma \ref{lem:cmp:equivalence} the diffferent sets are
  equivalent. To show that the construction captures all possible sets
  within this equivalence class and  according to Theorem~$4.4$ of
  Ref.~\onlinecite{Bandyopadhyay:2002lm}, we have to be able to apply
  an arbitrary symplectic matrix $f \in \mathrm{Sp}_{2m} (\F_2)$ to the
  generators  \eqref{eqn:semigroup:barGj}, so we get
  \begin{align}
    f G_0 = 
\begin{pmatrix} 
s\\ u
\end{pmatrix},
\qquad 
    f \bar G_j  = 
 \begin{pmatrix} 
s p_j(B) R + s A + t\\ 
  u p_j(B) R +    uA + v 
\end{pmatrix}
  \end{align}
  for
  \begin{align}
    f = 
\begin{pmatrix}
s & t\\
u & v
\end{pmatrix} ,
  \end{align}
  with $s,t,u,v \in M_m (\F_2)$.  We still are free to reset any of the
  generators to the $G_0$ by
  \begin{align}
    f^{\prime} = 
\begin{pmatrix}
s^{\prime}&t^{\prime}\\
u^{\prime}&v^{\prime}
\end{pmatrix},
  \end{align}
  which gives
  \begin{align}
    f^{\prime} f G_0 = 
\begin{pmatrix} 
s^{\prime}s+t^{\prime}u\\ 
u^{\prime}s + v^{\prime}u
\end{pmatrix}
  \end{align}
  and leads to the condition that $u's + v'u$. Applying $f^{\prime}$ to $f$
  leads then to a matrix where the lower left block equals $0_m$, thus
  we can chose a symplectic matrix in the form (\ref{eq:fuf}), 
  where $v\equiv (s^t)^{-1}$ as long as $f$ is symplectic. For the
  generators of the semigroup-based sets we get then the same result we
  found in Lemma \ref{lem:cmp:equivalence}.
\end{proof}

\section{Conclusions}
\label{sec:conclusion}

We have shown how to generate complete sets of cyclic MUBs with
different entanglement structures by taking advantage of the
properties of the Fibonacci polynomials. Two levels of generalization
raise the basic field-based sets with three completely factorizable
bases to the group-based sets with two completely factorizable bases
and finally to the semigroup-based sets with only one completely
factorizable basis.

Finally, we have proven that in this way we can generate all the
possible entanglement structures. Of course, sets with zero
factorizable bases are out of scope, as the standard basis cannot be
part of the set. In principle, this should be realizable by
generalizing again the discussed method.

\begin{acknowledgments}
  Financial support from the EU FP7 (Grant Q-ESSENCE) and the Spanish
  DGI (Grant FIS2011-26786) is gratefully acknowledged.
\end{acknowledgments}


\begin{thebibliography}{78}%
\makeatletter
\providecommand \@ifxundefined [1]{%
 \@ifx{#1\undefined}
}%
\providecommand \@ifnum [1]{%
 \ifnum #1\expandafter \@firstoftwo
 \else \expandafter \@secondoftwo
 \fi
}%
\providecommand \@ifx [1]{%
 \ifx #1\expandafter \@firstoftwo
 \else \expandafter \@secondoftwo
 \fi
}%
\providecommand \natexlab [1]{#1}%
\providecommand \enquote  [1]{``#1''}%
\providecommand \bibnamefont  [1]{#1}%
\providecommand \bibfnamefont [1]{#1}%
\providecommand \citenamefont [1]{#1}%
\providecommand \href@noop [0]{\@secondoftwo}%
\providecommand \href [0]{\begingroup \@sanitize@url \@href}%
\providecommand \@href[1]{\@@startlink{#1}\@@href}%
\providecommand \@@href[1]{\endgroup#1\@@endlink}%
\providecommand \@sanitize@url [0]{\catcode `\\12\catcode `\$12\catcode
  `\&12\catcode `\#12\catcode `\^12\catcode `\_12\catcode `\%12\relax}%
\providecommand \@@startlink[1]{}%
\providecommand \@@endlink[0]{}%
\providecommand \url  [0]{\begingroup\@sanitize@url \@url }%
\providecommand \@url [1]{\endgroup\@href {#1}{\urlprefix }}%
\providecommand \urlprefix  [0]{URL }%
\providecommand \Eprint [0]{\href }%
\providecommand \doibase [0]{http://dx.doi.org/}%
\providecommand \selectlanguage [0]{\@gobble}%
\providecommand \bibinfo  [0]{\@secondoftwo}%
\providecommand \bibfield  [0]{\@secondoftwo}%
\providecommand \translation [1]{[#1]}%
\providecommand \BibitemOpen [0]{}%
\providecommand \bibitemStop [0]{}%
\providecommand \bibitemNoStop [0]{.\EOS\space}%
\providecommand \EOS [0]{\spacefactor3000\relax}%
\providecommand \BibitemShut  [1]{\csname bibitem#1\endcsname}%
\let\auto@bib@innerbib\@empty
\bibitem [{\citenamefont {Bohr}(1928)}]{Bohr:1928fk}%
  \BibitemOpen
  \bibfield  {author} {\bibinfo {author} {\bibfnamefont {N.}~\bibnamefont
  {Bohr}},\ }\bibfield  {title} {\enquote {\bibinfo {title} {Das
  {Q}uantenpostulat und die neuere {E}ntwicklung der {A}tomistik},}\
  }\href@noop {} {\bibfield  {journal} {\bibinfo  {journal}
  {Naturwissenschaften}\ }\textbf {\bibinfo {volume} {16}},\ \bibinfo {pages}
  {245--248} (\bibinfo {year} {1928})}\BibitemShut {NoStop}%
\bibitem [{\citenamefont {Schwinger}(1960{\natexlab{a}})}]{Schwinger:1960a}%
  \BibitemOpen
  \bibfield  {author} {\bibinfo {author} {\bibfnamefont {J.}~\bibnamefont
  {Schwinger}},\ }\bibfield  {title} {\enquote {\bibinfo {title} {Unitary
  operator basis},}\ }\href@noop {} {\bibfield  {journal} {\bibinfo  {journal}
  {Proc. Natl. Acad. Sci. USA}\ }\textbf {\bibinfo {volume} {46}},\ \bibinfo
  {pages} {570--576} (\bibinfo {year} {1960}{\natexlab{a}})}\BibitemShut
  {NoStop}%
\bibitem [{\citenamefont {Schwinger}(1960{\natexlab{b}})}]{Schwinger:1960b}%
  \BibitemOpen
  \bibfield  {author} {\bibinfo {author} {\bibfnamefont {J.}~\bibnamefont
  {Schwinger}},\ }\bibfield  {title} {\enquote {\bibinfo {title} {Unitary
  transformations and the action principle},}\ }\href@noop {} {\bibfield
  {journal} {\bibinfo  {journal} {Proc. Natl. Acad. Sci. USA}\ }\textbf
  {\bibinfo {volume} {46}},\ \bibinfo {pages} {883--897} (\bibinfo {year}
  {1960}{\natexlab{b}})}\BibitemShut {NoStop}%
\bibitem [{\citenamefont {Schwinger}(1960{\natexlab{c}})}]{Schwinger:1960c}%
  \BibitemOpen
  \bibfield  {author} {\bibinfo {author} {\bibfnamefont {J.}~\bibnamefont
  {Schwinger}},\ }\bibfield  {title} {\enquote {\bibinfo {title} {The special
  canonical group},}\ }\href@noop {} {\bibfield  {journal} {\bibinfo  {journal}
  {Proc. Natl. Acad. Sci. USA}\ }\textbf {\bibinfo {volume} {46}},\ \bibinfo
  {pages} {1401--1415} (\bibinfo {year} {1960}{\natexlab{c}})}\BibitemShut
  {NoStop}%
\bibitem [{\citenamefont {Durt}\ \emph {et~al.}(2010)\citenamefont {Durt},
  \citenamefont {Englert}, \citenamefont {Bengtsson},\ and\ \citenamefont
  {\.{Z}yczkowski}}]{Durt:2010hc}%
  \BibitemOpen
  \bibfield  {author} {\bibinfo {author} {\bibfnamefont {T.}~\bibnamefont
  {Durt}}, \bibinfo {author} {\bibfnamefont {B.-G.}\ \bibnamefont {Englert}},
  \bibinfo {author} {\bibfnamefont {I.}~\bibnamefont {Bengtsson}}, \ and\
  \bibinfo {author} {\bibfnamefont {K.}~\bibnamefont {\.{Z}yczkowski}},\
  }\bibfield  {title} {\enquote {\bibinfo {title} {On mutually unbiased
  bases},}\ }\href@noop {} {\bibfield  {journal} {\bibinfo  {journal} {Int. J.
  Quantum Inf.}\ }\textbf {\bibinfo {volume} {8}},\ \bibinfo {pages} {535--640}
  (\bibinfo {year} {2010})}\BibitemShut {NoStop}%
\bibitem [{\citenamefont {Wootters}(1986)}]{Wootters:1986xy}%
  \BibitemOpen
  \bibfield  {author} {\bibinfo {author} {\bibfnamefont {W.~K.}\ \bibnamefont
  {Wootters}},\ }\bibfield  {title} {\enquote {\bibinfo {title} {Quantum
  mechanics without probability amplitudes},}\ }\href@noop {} {\bibfield
  {journal} {\bibinfo  {journal} {Found. Phys.}\ }\textbf {\bibinfo {volume}
  {16}},\ \bibinfo {pages} {391--405} (\bibinfo {year} {1986})}\BibitemShut
  {NoStop}%
\bibitem [{\citenamefont {Wootters}\ and\ \citenamefont
  {Fields}(1989)}]{Wootters:1989vn}%
  \BibitemOpen
  \bibfield  {author} {\bibinfo {author} {\bibfnamefont {W.~K.}\ \bibnamefont
  {Wootters}}\ and\ \bibinfo {author} {\bibfnamefont {B.~D.}\ \bibnamefont
  {Fields}},\ }\bibfield  {title} {\enquote {\bibinfo {title} {Optimal
  state-determination by mutually unbiased measurements},}\ }\href@noop {}
  {\bibfield  {journal} {\bibinfo  {journal} {Ann. Phys.}\ }\textbf {\bibinfo
  {volume} {191}},\ \bibinfo {pages} {363--381} (\bibinfo {year}
  {1989})}\BibitemShut {NoStop}%
\bibitem [{\citenamefont {Wootters}(1987)}]{Wootters:1987qf}%
  \BibitemOpen
  \bibfield  {author} {\bibinfo {author} {\bibfnamefont {W.~K.}\ \bibnamefont
  {Wootters}},\ }\bibfield  {title} {\enquote {\bibinfo {title} {A
  {W}igner-function formulation of finite-state quantum mechanics},}\
  }\href@noop {} {\bibfield  {journal} {\bibinfo  {journal} {Ann. Phys.}\
  }\textbf {\bibinfo {volume} {176}},\ \bibinfo {pages} {1--21} (\bibinfo
  {year} {1987})}\BibitemShut {NoStop}%
\bibitem [{\citenamefont {Asplund}\ and\ \citenamefont
  {Bj{\"{o}}rk}(2001)}]{Asplund:2001}%
  \BibitemOpen
  \bibfield  {author} {\bibinfo {author} {\bibfnamefont {R.}~\bibnamefont
  {Asplund}}\ and\ \bibinfo {author} {\bibfnamefont {G.}~\bibnamefont
  {Bj{\"{o}}rk}},\ }\bibfield  {title} {\enquote {\bibinfo {title}
  {Reconstructing the discrete {W}igner function and some properties of the
  measurement bases},}\ }\href {\doibase 10.1103/PhysRevA.64.012106} {\bibfield
   {journal} {\bibinfo  {journal} {Phys. Rev. A}\ }\textbf {\bibinfo {volume}
  {64}},\ \bibinfo {pages} {012106} (\bibinfo {year} {2001})}\BibitemShut
  {NoStop}%
\bibitem [{\citenamefont {Gibbons}, \citenamefont {Hoffman},\ and\
  \citenamefont {Wootters}(2004)}]{Gibbons:2004bh}%
  \BibitemOpen
  \bibfield  {author} {\bibinfo {author} {\bibfnamefont {K.~S.}\ \bibnamefont
  {Gibbons}}, \bibinfo {author} {\bibfnamefont {M.~J.}\ \bibnamefont
  {Hoffman}}, \ and\ \bibinfo {author} {\bibfnamefont {W.~K.}\ \bibnamefont
  {Wootters}},\ }\bibfield  {title} {\enquote {\bibinfo {title} {Discrete phase
  space based on finite fields},}\ }\href@noop {} {\bibfield  {journal}
  {\bibinfo  {journal} {Phys. Rev. A}\ }\textbf {\bibinfo {volume} {70}},\
  \bibinfo {pages} {062101} (\bibinfo {year} {2004})}\BibitemShut {NoStop}%
\bibitem [{\citenamefont {Vourdas}(2004)}]{Vourdas:2004nr}%
  \BibitemOpen
  \bibfield  {author} {\bibinfo {author} {\bibfnamefont {A.}~\bibnamefont
  {Vourdas}},\ }\bibfield  {title} {\enquote {\bibinfo {title} {Quantum systems
  with finite {H}ilbert space},}\ }\href@noop {} {\bibfield  {journal}
  {\bibinfo  {journal} {Rep. Prog. Phys.}\ }\textbf {\bibinfo {volume} {67}},\
  \bibinfo {pages} {267--320} (\bibinfo {year} {2004})}\BibitemShut {NoStop}%
\bibitem [{\citenamefont {Bj{\"{o}}rk}, \citenamefont {Klimov},\ and\
  \citenamefont {S{\'{a}}nchez-Soto}(2008)}]{Bjork:2008fv}%
  \BibitemOpen
  \bibfield  {author} {\bibinfo {author} {\bibfnamefont {G.}~\bibnamefont
  {Bj{\"{o}}rk}}, \bibinfo {author} {\bibfnamefont {A.~B.}\ \bibnamefont
  {Klimov}}, \ and\ \bibinfo {author} {\bibfnamefont {L.~L.}\ \bibnamefont
  {S{\'{a}}nchez-Soto}},\ }\bibfield  {title} {\enquote {\bibinfo {title} {The
  discrete {W}igner function},}\ }\href@noop {} {\bibfield  {journal} {\bibinfo
   {journal} {Prog. Opt.}\ }\textbf {\bibinfo {volume} {51}},\ \bibinfo {pages}
  {469--516} (\bibinfo {year} {2008})}\BibitemShut {NoStop}%
\bibitem [{\citenamefont {Bechmann-Pasquinucci}\ and\ \citenamefont
  {Peres}(2000)}]{Bechmann:2000ly}%
  \BibitemOpen
  \bibfield  {author} {\bibinfo {author} {\bibfnamefont {H.}~\bibnamefont
  {Bechmann-Pasquinucci}}\ and\ \bibinfo {author} {\bibfnamefont
  {A.}~\bibnamefont {Peres}},\ }\bibfield  {title} {\enquote {\bibinfo {title}
  {Quantum cryptography with 3-state systems},}\ }\href@noop {} {\bibfield
  {journal} {\bibinfo  {journal} {Phys. Rev. Lett.}\ }\textbf {\bibinfo
  {volume} {85}},\ \bibinfo {pages} {3313--3316} (\bibinfo {year}
  {2000})}\BibitemShut {NoStop}%
\bibitem [{\citenamefont {Bourennane}, \citenamefont {Karlsson},\ and\
  \citenamefont {Bj{\"o}rk}(2001)}]{Bourennane:2001qd}%
  \BibitemOpen
  \bibfield  {author} {\bibinfo {author} {\bibfnamefont {M.}~\bibnamefont
  {Bourennane}}, \bibinfo {author} {\bibfnamefont {A.}~\bibnamefont
  {Karlsson}}, \ and\ \bibinfo {author} {\bibfnamefont {G.}~\bibnamefont
  {Bj{\"o}rk}},\ }\bibfield  {title} {\enquote {\bibinfo {title} {Quantum key
  distribution using multilevel encoding},}\ }\href@noop {} {\bibfield
  {journal} {\bibinfo  {journal} {Phys. Rev. A}\ }\textbf {\bibinfo {volume}
  {64}},\ \bibinfo {pages} {012306} (\bibinfo {year} {2001})}\BibitemShut
  {NoStop}%
\bibitem [{\citenamefont {Nikolopoulos}\ and\ \citenamefont
  {Alber}(2005)}]{Nikolopoulos:2005yg}%
  \BibitemOpen
  \bibfield  {author} {\bibinfo {author} {\bibfnamefont {G.~M.}\ \bibnamefont
  {Nikolopoulos}}\ and\ \bibinfo {author} {\bibfnamefont {G.}~\bibnamefont
  {Alber}},\ }\bibfield  {title} {\enquote {\bibinfo {title} {Security bound of
  two-basis quantum-key-distribution protocols using qudits},}\ }\href@noop {}
  {\bibfield  {journal} {\bibinfo  {journal} {Phys. Rev. A}\ }\textbf {\bibinfo
  {volume} {72}},\ \bibinfo {pages} {032320} (\bibinfo {year}
  {2005})}\BibitemShut {NoStop}%
\bibitem [{\citenamefont {Yu}, \citenamefont {Lin},\ and\ \citenamefont
  {Huang}(2008)}]{Yu:2008wd}%
  \BibitemOpen
  \bibfield  {author} {\bibinfo {author} {\bibfnamefont {I.-C.}\ \bibnamefont
  {Yu}}, \bibinfo {author} {\bibfnamefont {F.-L.}\ \bibnamefont {Lin}}, \ and\
  \bibinfo {author} {\bibfnamefont {C.-Y.}\ \bibnamefont {Huang}},\ }\bibfield
  {title} {\enquote {\bibinfo {title} {Quantum secret sharing with multilevel
  mutually (un)biased bases},}\ }\href@noop {} {\bibfield  {journal} {\bibinfo
  {journal} {Phys. Rev. A}\ }\textbf {\bibinfo {volume} {78}},\ \bibinfo
  {pages} {012344} (\bibinfo {year} {2008})}\BibitemShut {NoStop}%
\bibitem [{\citenamefont {Xiong}\ \emph {et~al.}(2012)\citenamefont {Xiong},
  \citenamefont {Shi}, \citenamefont {Wang}, \citenamefont {Jing},
  \citenamefont {Lei}, \citenamefont {Mu},\ and\ \citenamefont
  {Fan}}]{Xiong:2012fe}%
  \BibitemOpen
  \bibfield  {author} {\bibinfo {author} {\bibfnamefont {Z.-X.}\ \bibnamefont
  {Xiong}}, \bibinfo {author} {\bibfnamefont {H.-D.}\ \bibnamefont {Shi}},
  \bibinfo {author} {\bibfnamefont {Y.-N.}\ \bibnamefont {Wang}}, \bibinfo
  {author} {\bibfnamefont {L.}~\bibnamefont {Jing}}, \bibinfo {author}
  {\bibfnamefont {J.}~\bibnamefont {Lei}}, \bibinfo {author} {\bibfnamefont
  {L.-Z.}\ \bibnamefont {Mu}}, \ and\ \bibinfo {author} {\bibfnamefont
  {H.}~\bibnamefont {Fan}},\ }\bibfield  {title} {\enquote {\bibinfo {title}
  {General quantum key distribution in higher dimension},}\ }\href@noop {}
  {\bibfield  {journal} {\bibinfo  {journal} {Phys. Rev. A}\ }\textbf {\bibinfo
  {volume} {85}},\ \bibinfo {pages} {012334} (\bibinfo {year}
  {2012})}\BibitemShut {NoStop}%
\bibitem [{\citenamefont {Mafu}\ \emph {et~al.}(2013)\citenamefont {Mafu},
  \citenamefont {Dudley}, \citenamefont {Goyal}, \citenamefont {Giovannini},
  \citenamefont {McLaren}, \citenamefont {Padgett}, \citenamefont {Konrad},
  \citenamefont {Petruccione}, \citenamefont {L{\"u}tkenhaus},\ and\
  \citenamefont {Forbes}}]{Mafu:2013qc}%
  \BibitemOpen
  \bibfield  {author} {\bibinfo {author} {\bibfnamefont {M.}~\bibnamefont
  {Mafu}}, \bibinfo {author} {\bibfnamefont {A.}~\bibnamefont {Dudley}},
  \bibinfo {author} {\bibfnamefont {S.}~\bibnamefont {Goyal}}, \bibinfo
  {author} {\bibfnamefont {D.}~\bibnamefont {Giovannini}}, \bibinfo {author}
  {\bibfnamefont {M.}~\bibnamefont {McLaren}}, \bibinfo {author} {\bibfnamefont
  {M.~J.}\ \bibnamefont {Padgett}}, \bibinfo {author} {\bibfnamefont
  {T.}~\bibnamefont {Konrad}}, \bibinfo {author} {\bibfnamefont
  {F.}~\bibnamefont {Petruccione}}, \bibinfo {author} {\bibfnamefont
  {N.}~\bibnamefont {L{\"u}tkenhaus}}, \ and\ \bibinfo {author} {\bibfnamefont
  {A.}~\bibnamefont {Forbes}},\ }\bibfield  {title} {\enquote {\bibinfo {title}
  {Higher-dimensional orbital-angular-momentum-based quantum key distribution
  with mutually unbiased bases},}\ }\href@noop {} {\bibfield  {journal}
  {\bibinfo  {journal} {Phys. Rev. A}\ }\textbf {\bibinfo {volume} {88}},\
  \bibinfo {pages} {032305} (\bibinfo {year} {2013})}\BibitemShut {NoStop}%
\bibitem [{\citenamefont {Gottesman}(1996)}]{Gottesman:1996ve}%
  \BibitemOpen
  \bibfield  {author} {\bibinfo {author} {\bibfnamefont {D.}~\bibnamefont
  {Gottesman}},\ }\bibfield  {title} {\enquote {\bibinfo {title} {Class of
  quantum error-correcting codes saturating the quantum {H}amming bound},}\
  }\href@noop {} {\bibfield  {journal} {\bibinfo  {journal} {Phys. Rev. A}\
  }\textbf {\bibinfo {volume} {54}},\ \bibinfo {pages} {1862--1868} (\bibinfo
  {year} {1996})}\BibitemShut {NoStop}%
\bibitem [{\citenamefont {Calderbank}\ and\ \citenamefont
  {Shor}(1996)}]{Calderbank:1996yo}%
  \BibitemOpen
  \bibfield  {author} {\bibinfo {author} {\bibfnamefont {A.~R.}\ \bibnamefont
  {Calderbank}}\ and\ \bibinfo {author} {\bibfnamefont {P.~W.}\ \bibnamefont
  {Shor}},\ }\bibfield  {title} {\enquote {\bibinfo {title} {Good quantum
  error-correcting codes exist},}\ }\href@noop {} {\bibfield  {journal}
  {\bibinfo  {journal} {Phys. Rev. A}\ }\textbf {\bibinfo {volume} {54}},\
  \bibinfo {pages} {1098--1105} (\bibinfo {year} {1996})}\BibitemShut {NoStop}%
\bibitem [{\citenamefont {Calderbank}\ \emph
  {et~al.}(1997{\natexlab{a}})\citenamefont {Calderbank}, \citenamefont
  {Rains}, \citenamefont {Shor},\ and\ \citenamefont
  {Sloane}}]{Calderbank:1997qf}%
  \BibitemOpen
  \bibfield  {author} {\bibinfo {author} {\bibfnamefont {A.~R.}\ \bibnamefont
  {Calderbank}}, \bibinfo {author} {\bibfnamefont {E.~M.}\ \bibnamefont
  {Rains}}, \bibinfo {author} {\bibfnamefont {P.~W.}\ \bibnamefont {Shor}}, \
  and\ \bibinfo {author} {\bibfnamefont {N.~J.~A.}\ \bibnamefont {Sloane}},\
  }\bibfield  {title} {\enquote {\bibinfo {title} {Quantum error correction and
  orthogonal geometry},}\ }\href@noop {} {\bibfield  {journal} {\bibinfo
  {journal} {Phys. Rev. Lett.}\ }\textbf {\bibinfo {volume} {78}},\ \bibinfo
  {pages} {405--408} (\bibinfo {year} {1997}{\natexlab{a}})}\BibitemShut
  {NoStop}%
\bibitem [{\citenamefont {Griffiths}(2005)}]{Griffiths:2005rf}%
  \BibitemOpen
  \bibfield  {author} {\bibinfo {author} {\bibfnamefont {R.~B.}\ \bibnamefont
  {Griffiths}},\ }\bibfield  {title} {\enquote {\bibinfo {title} {Channel kets,
  entangled states, and the location of quantum information},}\ }\href@noop {}
  {\bibfield  {journal} {\bibinfo  {journal} {Phys. Rev. A}\ }\textbf {\bibinfo
  {volume} {71}},\ \bibinfo {pages} {042337} (\bibinfo {year}
  {2005})}\BibitemShut {NoStop}%
\bibitem [{\citenamefont {Paw{\l}owski}\ and\ \citenamefont
  {{\.Z}ukowski}(2010)}]{Pawowski:2010if}%
  \BibitemOpen
  \bibfield  {author} {\bibinfo {author} {\bibfnamefont {M.}~\bibnamefont
  {Paw{\l}owski}}\ and\ \bibinfo {author} {\bibfnamefont {M.}~\bibnamefont
  {{\.Z}ukowski}},\ }\bibfield  {title} {\enquote {\bibinfo {title}
  {Entanglement-assisted random access codes},}\ }\href@noop {} {\bibfield
  {journal} {\bibinfo  {journal} {Phys. Rev. A}\ }\textbf {\bibinfo {volume}
  {81}},\ \bibinfo {pages} {042326} (\bibinfo {year} {2010})}\BibitemShut
  {NoStop}%
\bibitem [{\citenamefont {Aharonov}\ and\ \citenamefont
  {Englert}(2001)}]{Aharonov:2001dq}%
  \BibitemOpen
  \bibfield  {author} {\bibinfo {author} {\bibfnamefont {Y.}~\bibnamefont
  {Aharonov}}\ and\ \bibinfo {author} {\bibfnamefont {B.-G.}\ \bibnamefont
  {Englert}},\ }\bibfield  {title} {\enquote {\bibinfo {title} {The mean king's
  problem: spin 1},}\ }\href@noop {} {\bibfield  {journal} {\bibinfo  {journal}
  {Z. Naturforsch.}\ }\textbf {\bibinfo {volume} {56a}},\ \bibinfo {pages}
  {16--19} (\bibinfo {year} {2001})}\BibitemShut {NoStop}%
\bibitem [{\citenamefont {Englert}\ and\ \citenamefont
  {Aharonov}(2001)}]{Englert:2001cr}%
  \BibitemOpen
  \bibfield  {author} {\bibinfo {author} {\bibfnamefont {B.-G.}\ \bibnamefont
  {Englert}}\ and\ \bibinfo {author} {\bibfnamefont {Y.}~\bibnamefont
  {Aharonov}},\ }\bibfield  {title} {\enquote {\bibinfo {title} {The mean
  king's problem: Prime degrees of freedom},}\ }\href@noop {} {\bibfield
  {journal} {\bibinfo  {journal} {Phys. Lett. A}\ }\textbf {\bibinfo {volume}
  {284}},\ \bibinfo {pages} {1--5} (\bibinfo {year} {2001})}\BibitemShut
  {NoStop}%
\bibitem [{\citenamefont {Aravind}(2003)}]{Aravind:2003nx}%
  \BibitemOpen
  \bibfield  {author} {\bibinfo {author} {\bibfnamefont {P.~K.}\ \bibnamefont
  {Aravind}},\ }\bibfield  {title} {\enquote {\bibinfo {title} {Solution to the
  king's problem in prime power dimensions},}\ }\href@noop {} {\bibfield
  {journal} {\bibinfo  {journal} {Z. Naturforsch.}\ }\textbf {\bibinfo {volume}
  {58a}},\ \bibinfo {pages} {85--92} (\bibinfo {year} {2003})}\BibitemShut
  {NoStop}%
\bibitem [{\citenamefont {Klappenecker}\ and\ \citenamefont
  {R{\"{o}}tteler}(2005)}]{Klappenecker:2005lf}%
  \BibitemOpen
  \bibfield  {author} {\bibinfo {author} {\bibfnamefont {A.}~\bibnamefont
  {Klappenecker}}\ and\ \bibinfo {author} {\bibfnamefont {M.}~\bibnamefont
  {R{\"{o}}tteler}},\ }\bibfield  {title} {\enquote {\bibinfo {title} {New
  tales of the mean king},}\ }\href@noop {} {\bibfield  {journal} {\bibinfo
  {journal} {arXiv: quant-ph/0502138}\ } (\bibinfo {year} {2005})}\BibitemShut
  {NoStop}%
\bibitem [{\citenamefont {Hayashi}, \citenamefont {Horibe},\ and\ \citenamefont
  {Hashimoto}(2005)}]{Hayashi:2005oq}%
  \BibitemOpen
  \bibfield  {author} {\bibinfo {author} {\bibfnamefont {A.}~\bibnamefont
  {Hayashi}}, \bibinfo {author} {\bibfnamefont {M.}~\bibnamefont {Horibe}}, \
  and\ \bibinfo {author} {\bibfnamefont {T.}~\bibnamefont {Hashimoto}},\
  }\bibfield  {title} {\enquote {\bibinfo {title} {Mean king's problem with
  mutually unbiased bases and orthogonal {L}atin squares},}\ }\href@noop {}
  {\bibfield  {journal} {\bibinfo  {journal} {Phys. Rev. A}\ }\textbf {\bibinfo
  {volume} {71}},\ \bibinfo {pages} {052331} (\bibinfo {year}
  {2005})}\BibitemShut {NoStop}%
\bibitem [{\citenamefont {Paz}, \citenamefont {Roncaglia},\ and\ \citenamefont
  {Saraceno}(2005)}]{Paz:2005dz}%
  \BibitemOpen
  \bibfield  {author} {\bibinfo {author} {\bibfnamefont {J.~P.}\ \bibnamefont
  {Paz}}, \bibinfo {author} {\bibfnamefont {A.~J.}\ \bibnamefont {Roncaglia}},
  \ and\ \bibinfo {author} {\bibfnamefont {M.}~\bibnamefont {Saraceno}},\
  }\bibfield  {title} {\enquote {\bibinfo {title} {Qubits in phase space:
  {W}igner-function approach to quantum-error correction and the mean-king
  problem},}\ }\href {\doibase 10.1103/PhysRevA.72.012309} {\bibfield
  {journal} {\bibinfo  {journal} {Phys. Rev. A}\ }\textbf {\bibinfo {volume}
  {72}},\ \bibinfo {pages} {012309} (\bibinfo {year} {2005})}\BibitemShut
  {NoStop}%
\bibitem [{\citenamefont {Durt}(2006)}]{Durt:2006fu}%
  \BibitemOpen
  \bibfield  {author} {\bibinfo {author} {\bibfnamefont {T.}~\bibnamefont
  {Durt}},\ }\bibfield  {title} {\enquote {\bibinfo {title} {About {W}eyl and
  {W}igner tomography in finite-dimensional {H}ilbert spaces},}\ }\href@noop {}
  {\bibfield  {journal} {\bibinfo  {journal} {Open Sys. Inf. Dyn.}\ }\textbf
  {\bibinfo {volume} {13}},\ \bibinfo {pages} {403--413} (\bibinfo {year}
  {2006})}\BibitemShut {NoStop}%
\bibitem [{\citenamefont {Kimura}, \citenamefont {Tanaka},\ and\ \citenamefont
  {Ozawa}(2006)}]{Kimura:2006kl}%
  \BibitemOpen
  \bibfield  {author} {\bibinfo {author} {\bibfnamefont {G.}~\bibnamefont
  {Kimura}}, \bibinfo {author} {\bibfnamefont {H.}~\bibnamefont {Tanaka}}, \
  and\ \bibinfo {author} {\bibfnamefont {M.}~\bibnamefont {Ozawa}},\ }\bibfield
   {title} {\enquote {\bibinfo {title} {Solution to the mean king's problem
  with mutually unbiased bases for arbitrary levels},}\ }\href {\doibase
  10.1103/PhysRevA.73.050301} {\bibfield  {journal} {\bibinfo  {journal} {Phys.
  Rev. A}\ }\textbf {\bibinfo {volume} {73}},\ \bibinfo {pages} {050301(R)}
  (\bibinfo {year} {2006})}\BibitemShut {NoStop}%
\bibitem [{\citenamefont {Kalev}, \citenamefont {Mann},\ and\ \citenamefont
  {Revzen}(2013)}]{Kalev:2013bx}%
  \BibitemOpen
  \bibfield  {author} {\bibinfo {author} {\bibfnamefont {A.}~\bibnamefont
  {Kalev}}, \bibinfo {author} {\bibfnamefont {A.}~\bibnamefont {Mann}}, \ and\
  \bibinfo {author} {\bibfnamefont {M.}~\bibnamefont {Revzen}},\ }\bibfield
  {title} {\enquote {\bibinfo {title} {Quantum-mechanical retrodiction through
  an extended mean king problem},}\ }\href@noop {} {\bibfield  {journal}
  {\bibinfo  {journal} {EPL}\ }\textbf {\bibinfo {volume} {104}},\ \bibinfo
  {pages} {50008} (\bibinfo {year} {2013})}\BibitemShut {NoStop}%
\bibitem [{\citenamefont {Revzen}(2013)}]{Revzen:2013ax}%
  \BibitemOpen
  \bibfield  {author} {\bibinfo {author} {\bibfnamefont {M.}~\bibnamefont
  {Revzen}},\ }\bibfield  {title} {\enquote {\bibinfo {title} {Maximal
  entanglement, collective coordinates and tracking the king},}\ }\href@noop {}
  {\bibfield  {journal} {\bibinfo  {journal} {J. Phys. A}\ }\textbf {\bibinfo
  {volume} {46}},\ \bibinfo {pages} {075303} (\bibinfo {year}
  {2013})}\BibitemShut {NoStop}%
\bibitem [{\citenamefont {Ivanovic}(1981)}]{Ivanovic:1981tg}%
  \BibitemOpen
  \bibfield  {author} {\bibinfo {author} {\bibfnamefont {I.~D.}\ \bibnamefont
  {Ivanovic}},\ }\bibfield  {title} {\enquote {\bibinfo {title} {Geometrical
  description of quantal state determination},}\ }\href@noop {} {\bibfield
  {journal} {\bibinfo  {journal} {J. Phys. A}\ }\textbf {\bibinfo {volume}
  {14}},\ \bibinfo {pages} {3241--3245} (\bibinfo {year} {1981})}\BibitemShut
  {NoStop}%
\bibitem [{\citenamefont {Calderbank}\ \emph
  {et~al.}(1997{\natexlab{b}})\citenamefont {Calderbank}, \citenamefont
  {Cameron}, \citenamefont {Kantor},\ and\ \citenamefont
  {Seidel}}]{Calderbank:1997ao}%
  \BibitemOpen
  \bibfield  {author} {\bibinfo {author} {\bibfnamefont {A.~R.}\ \bibnamefont
  {Calderbank}}, \bibinfo {author} {\bibfnamefont {P.~J.}\ \bibnamefont
  {Cameron}}, \bibinfo {author} {\bibfnamefont {W.~M.}\ \bibnamefont {Kantor}},
  \ and\ \bibinfo {author} {\bibfnamefont {J.~J.}\ \bibnamefont {Seidel}},\
  }\bibfield  {title} {\enquote {\bibinfo {title} {$\mathbb{Z}_4$-kerdock
  codes, orthogonal spreads, and extremal {E}uclidean line-sets},}\ }\href@noop
  {} {\bibfield  {journal} {\bibinfo  {journal} {Proc. London Math. Soc.}\
  }\textbf {\bibinfo {volume} {75}},\ \bibinfo {pages} {436--480} (\bibinfo
  {year} {1997}{\natexlab{b}})}\BibitemShut {NoStop}%
\bibitem [{\citenamefont {Grassl}(2004)}]{Grassl:2004iz}%
  \BibitemOpen
  \bibfield  {author} {\bibinfo {author} {\bibfnamefont {M.}~\bibnamefont
  {Grassl}},\ }\bibfield  {title} {\enquote {\bibinfo {title} {On {SIC-POVM}s
  and {MUB}s in dimension 6},}\ }\href@noop {} {\bibfield  {journal} {\bibinfo
  {journal} {arXiv: quant-ph/0406175}\ } (\bibinfo {year} {2004})}\BibitemShut
  {NoStop}%
\bibitem [{\citenamefont {Butterley}\ and\ \citenamefont
  {Hall}(2007)}]{Butterley:2007fe}%
  \BibitemOpen
  \bibfield  {author} {\bibinfo {author} {\bibfnamefont {P.}~\bibnamefont
  {Butterley}}\ and\ \bibinfo {author} {\bibfnamefont {W.}~\bibnamefont
  {Hall}},\ }\bibfield  {title} {\enquote {\bibinfo {title} {Numerical evidence
  for the maximum number of mutually unbiased bases in dimension six},}\
  }\href@noop {} {\bibfield  {journal} {\bibinfo  {journal} {Phys. Lett. A}\
  }\textbf {\bibinfo {volume} {369}},\ \bibinfo {pages} {5--8} (\bibinfo {year}
  {2007})}\BibitemShut {NoStop}%
\bibitem [{\citenamefont {Brierley}\ and\ \citenamefont
  {Weigert}(2008)}]{Brierley:2008ja}%
  \BibitemOpen
  \bibfield  {author} {\bibinfo {author} {\bibfnamefont {S.}~\bibnamefont
  {Brierley}}\ and\ \bibinfo {author} {\bibfnamefont {S.}~\bibnamefont
  {Weigert}},\ }\bibfield  {title} {\enquote {\bibinfo {title} {Maximal sets of
  mutually unbiased quantum states in dimension 6},}\ }\href@noop {} {\bibfield
   {journal} {\bibinfo  {journal} {Phys. Rev. A}\ }\textbf {\bibinfo {volume}
  {78}},\ \bibinfo {pages} {042312} (\bibinfo {year} {2008})}\BibitemShut
  {NoStop}%
\bibitem [{\citenamefont {Brierley}\ and\ \citenamefont
  {Weigert}(2009)}]{Brierley:2009mw}%
  \BibitemOpen
  \bibfield  {author} {\bibinfo {author} {\bibfnamefont {S.}~\bibnamefont
  {Brierley}}\ and\ \bibinfo {author} {\bibfnamefont {S.}~\bibnamefont
  {Weigert}},\ }\bibfield  {title} {\enquote {\bibinfo {title} {Constructing
  mutually unbiased bases in dimension six},}\ }\href@noop {} {\bibfield
  {journal} {\bibinfo  {journal} {Phys. Rev. A}\ }\textbf {\bibinfo {volume}
  {79}},\ \bibinfo {pages} {052316} (\bibinfo {year} {2009})}\BibitemShut
  {NoStop}%
\bibitem [{\citenamefont {Raynal}, \citenamefont {L{\"u}},\ and\ \citenamefont
  {Englert}(2011)}]{Raynal:2011xc}%
  \BibitemOpen
  \bibfield  {author} {\bibinfo {author} {\bibfnamefont {P.}~\bibnamefont
  {Raynal}}, \bibinfo {author} {\bibfnamefont {X.}~\bibnamefont {L{\"u}}}, \
  and\ \bibinfo {author} {\bibfnamefont {B.-G.}\ \bibnamefont {Englert}},\
  }\bibfield  {title} {\enquote {\bibinfo {title} {Mutually unbiased bases in
  six dimensions: The four most distant bases},}\ }\href@noop {} {\bibfield
  {journal} {\bibinfo  {journal} {Phys. Rev. A}\ }\textbf {\bibinfo {volume}
  {83}},\ \bibinfo {pages} {062303} (\bibinfo {year} {2011})}\BibitemShut
  {NoStop}%
\bibitem [{\citenamefont {McNulty}\ and\ \citenamefont
  {Weigert}(2012)}]{McNulty:2012ez}%
  \BibitemOpen
  \bibfield  {author} {\bibinfo {author} {\bibfnamefont {D.}~\bibnamefont
  {McNulty}}\ and\ \bibinfo {author} {\bibfnamefont {S.}~\bibnamefont
  {Weigert}},\ }\bibfield  {title} {\enquote {\bibinfo {title} {All mutually
  unbiased product bases in dimension 6},}\ }\href@noop {} {\bibfield
  {journal} {\bibinfo  {journal} {J. Phys. A}\ }\textbf {\bibinfo {volume}
  {45}},\ \bibinfo {pages} {135307} (\bibinfo {year} {2012})}\BibitemShut
  {NoStop}%
\bibitem [{\citenamefont {Archer}(2005)}]{Archer:2005bz}%
  \BibitemOpen
  \bibfield  {author} {\bibinfo {author} {\bibfnamefont {C.}~\bibnamefont
  {Archer}},\ }\bibfield  {title} {\enquote {\bibinfo {title} {There is no
  generalization of known formulas for mutually unbiased bases},}\ }\href@noop
  {} {\bibfield  {journal} {\bibinfo  {journal} {J. Math. Phys.}\ }\textbf
  {\bibinfo {volume} {46}},\ \bibinfo {pages} {022106} (\bibinfo {year}
  {2005})}\BibitemShut {NoStop}%
\bibitem [{\citenamefont {Wocjian}\ and\ \citenamefont
  {Beth}(2005)}]{Wocjian:2005pw}%
  \BibitemOpen
  \bibfield  {author} {\bibinfo {author} {\bibfnamefont {P.}~\bibnamefont
  {Wocjian}}\ and\ \bibinfo {author} {\bibfnamefont {T.}~\bibnamefont {Beth}},\
  }\bibfield  {title} {\enquote {\bibinfo {title} {New construction of mutually
  unbiased basis in square dimensions},}\ }\href@noop {} {\bibfield  {journal}
  {\bibinfo  {journal} {Quantum Info. Compu.}\ }\textbf {\bibinfo {volume}
  {5}},\ \bibinfo {pages} {93--101} (\bibinfo {year} {2005})}\BibitemShut
  {NoStop}%
\bibitem [{\citenamefont {Saniga}, \citenamefont {Planat},\ and\ \citenamefont
  {Rosu}(2004)}]{Saniga:2004xa}%
  \BibitemOpen
  \bibfield  {author} {\bibinfo {author} {\bibfnamefont {M.}~\bibnamefont
  {Saniga}}, \bibinfo {author} {\bibfnamefont {M.}~\bibnamefont {Planat}}, \
  and\ \bibinfo {author} {\bibfnamefont {H.}~\bibnamefont {Rosu}},\ }\bibfield
  {title} {\enquote {\bibinfo {title} {Mutually unbiased bases and finite
  projective planes},}\ }\href@noop {} {\bibfield  {journal} {\bibinfo
  {journal} {J. Opt. B}\ }\textbf {\bibinfo {volume} {6}},\ \bibinfo {pages}
  {L19--L20} (\bibinfo {year} {2004})}\BibitemShut {NoStop}%
\bibitem [{\citenamefont {Bengtsson}(2004)}]{Bengtsson:2004ol}%
  \BibitemOpen
  \bibfield  {author} {\bibinfo {author} {\bibfnamefont {I.}~\bibnamefont
  {Bengtsson}},\ }\bibfield  {title} {\enquote {\bibinfo {title} {Mubs,
  polytopes and finite geometries},}\ }\href@noop {} {\bibfield  {journal}
  {\bibinfo  {journal} {arXiv: quant-ph/0406174}\ } (\bibinfo {year}
  {2004})}\BibitemShut {NoStop}%
\bibitem [{\citenamefont {Weigert}\ and\ \citenamefont
  {Durt}(2010)}]{Weigert:2010ev}%
  \BibitemOpen
  \bibfield  {author} {\bibinfo {author} {\bibfnamefont {S.}~\bibnamefont
  {Weigert}}\ and\ \bibinfo {author} {\bibfnamefont {T.}~\bibnamefont {Durt}},\
  }\bibfield  {title} {\enquote {\bibinfo {title} {Affine constellations
  without mutually unbiased counterparts},}\ }\href@noop {} {\bibfield
  {journal} {\bibinfo  {journal} {J. Phys. A}\ }\textbf {\bibinfo {volume}
  {43}},\ \bibinfo {pages} {402002} (\bibinfo {year} {2010})}\BibitemShut
  {NoStop}%
\bibitem [{\citenamefont {Zauner}(2011)}]{Zauner:2011fv}%
  \BibitemOpen
  \bibfield  {author} {\bibinfo {author} {\bibfnamefont {G.}~\bibnamefont
  {Zauner}},\ }\bibfield  {title} {\enquote {\bibinfo {title} {Quantum designs:
  Foundations of a noncommutative design theory},}\ }\href@noop {} {\bibfield
  {journal} {\bibinfo  {journal} {Int. J. Quantum Inf.}\ }\textbf {\bibinfo
  {volume} {9}},\ \bibinfo {pages} {445--507} (\bibinfo {year}
  {2011})}\BibitemShut {NoStop}%
\bibitem [{\citenamefont {Wootters}(2006)}]{Wootters:2006vk}%
  \BibitemOpen
  \bibfield  {author} {\bibinfo {author} {\bibfnamefont {W.~K.}\ \bibnamefont
  {Wootters}},\ }\bibfield  {title} {\enquote {\bibinfo {title} {Quantum
  measurements and finite geometry},}\ }\href@noop {} {\bibfield  {journal}
  {\bibinfo  {journal} {Found. Phys.}\ }\textbf {\bibinfo {volume} {36}},\
  \bibinfo {pages} {112--126} (\bibinfo {year} {2006})}\BibitemShut {NoStop}%
\bibitem [{\citenamefont {Paterek}, \citenamefont {Daki{\'c}},\ and\
  \citenamefont {Brukner}(2009)}]{Paterek:2009zs}%
  \BibitemOpen
  \bibfield  {author} {\bibinfo {author} {\bibfnamefont {T.}~\bibnamefont
  {Paterek}}, \bibinfo {author} {\bibfnamefont {B.}~\bibnamefont {Daki{\'c}}},
  \ and\ \bibinfo {author} {\bibfnamefont {{\v C}.}~\bibnamefont {Brukner}},\
  }\bibfield  {title} {\enquote {\bibinfo {title} {Mutually unbiased bases,
  orthogonal {L}atin squares, and hidden-variable models},}\ }\href@noop {}
  {\bibfield  {journal} {\bibinfo  {journal} {Phys. Rev. A}\ }\textbf {\bibinfo
  {volume} {79}},\ \bibinfo {pages} {012109} (\bibinfo {year}
  {2009})}\BibitemShut {NoStop}%
\bibitem [{\citenamefont {Hall}\ and\ \citenamefont {Rao}(2010)}]{Hall:2010ir}%
  \BibitemOpen
  \bibfield  {author} {\bibinfo {author} {\bibfnamefont {J.~L.}\ \bibnamefont
  {Hall}}\ and\ \bibinfo {author} {\bibfnamefont {A.}~\bibnamefont {Rao}},\
  }\bibfield  {title} {\enquote {\bibinfo {title} {Mutually orthogonal {L}atin
  squares from the inner products of vectors in mutually unbiased bases},}\
  }\href@noop {} {\bibfield  {journal} {\bibinfo  {journal} {J. Phys. A}\
  }\textbf {\bibinfo {volume} {43}},\ \bibinfo {pages} {135302} (\bibinfo
  {year} {2010})}\BibitemShut {NoStop}%
\bibitem [{\citenamefont {Appleby}(2009)}]{Appleby:2009vn}%
  \BibitemOpen
  \bibfield  {author} {\bibinfo {author} {\bibfnamefont {D.~M.}\ \bibnamefont
  {Appleby}},\ }\bibfield  {title} {\enquote {\bibinfo {title} {{SIC-POVM}s and
  {MUBs}: Geometrical relationships in prime dimension},}\ }in\ \href@noop {}
  {\emph {\bibinfo {booktitle} {Foundations of {P}robability and {P}hysics}}},\
  \bibinfo {series} {{AIP} Conf. Proc.}, Vol.~\bibinfo {volume} {5}\ (\bibinfo
  {publisher} {Amer. Inst. Phys.},\ \bibinfo {year} {2009})\ pp.\ \bibinfo
  {pages} {223---232}\BibitemShut {NoStop}%
\bibitem [{\citenamefont {Klappenecker}\ and\ \citenamefont
  {R{\"o}tteler}(2005)}]{Klappenecker:2005qd}%
  \BibitemOpen
  \bibfield  {author} {\bibinfo {author} {\bibfnamefont {A.}~\bibnamefont
  {Klappenecker}}\ and\ \bibinfo {author} {\bibfnamefont {M.}~\bibnamefont
  {R{\"o}tteler}},\ }\bibfield  {title} {\enquote {\bibinfo {title} {Mutually
  unbiased bases are complex projective 2-designs},}\ }in\ \href@noop {} {\emph
  {\bibinfo {booktitle} {Proc. 2005 IEEE International Symposium on Information
  Theory}}}\ (\bibinfo {address} {Adelaide},\ \bibinfo {year} {2005})\ pp.\
  \bibinfo {pages} {1740--1744}\BibitemShut {NoStop}%
\bibitem [{\citenamefont {Gross}, \citenamefont {Audenaert},\ and\
  \citenamefont {Eisert}(2007)}]{Gross:2007cs}%
  \BibitemOpen
  \bibfield  {author} {\bibinfo {author} {\bibfnamefont {D.}~\bibnamefont
  {Gross}}, \bibinfo {author} {\bibfnamefont {K.}~\bibnamefont {Audenaert}}, \
  and\ \bibinfo {author} {\bibfnamefont {J.}~\bibnamefont {Eisert}},\
  }\bibfield  {title} {\enquote {\bibinfo {title} {Evenly distributed
  unitaries: On the structure of unitary designs},}\ }\href@noop {} {\bibfield
  {journal} {\bibinfo  {journal} {J. Math. Phys.}\ }\textbf {\bibinfo {volume}
  {48}},\ \bibinfo {pages} {052104} (\bibinfo {year} {2007})}\BibitemShut
  {NoStop}%
\bibitem [{\citenamefont {Bandyopadhyay}\ \emph {et~al.}(2002)\citenamefont
  {Bandyopadhyay}, \citenamefont {Boykin}, \citenamefont {Roychowdhury},\ and\
  \citenamefont {Vatan}}]{Bandyopadhyay:2002lm}%
  \BibitemOpen
  \bibfield  {author} {\bibinfo {author} {\bibfnamefont {S.}~\bibnamefont
  {Bandyopadhyay}}, \bibinfo {author} {\bibfnamefont {P.~O.}\ \bibnamefont
  {Boykin}}, \bibinfo {author} {\bibfnamefont {V.}~\bibnamefont
  {Roychowdhury}}, \ and\ \bibinfo {author} {\bibfnamefont {F.}~\bibnamefont
  {Vatan}},\ }\bibfield  {title} {\enquote {\bibinfo {title} {A new proof for
  the existence of mutually unbiased bases},}\ }\href {\doibase
  10.1007/s00453-002-0980-7} {\bibfield  {journal} {\bibinfo  {journal}
  {Algorithmica}\ }\textbf {\bibinfo {volume} {34}},\ \bibinfo {pages}
  {512--528} (\bibinfo {year} {2002})}\BibitemShut {NoStop}%
\bibitem [{\citenamefont {Parthasarathy}(2004)}]{Parthasarathy:2004pf}%
  \BibitemOpen
  \bibfield  {author} {\bibinfo {author} {\bibfnamefont {K.~R.}\ \bibnamefont
  {Parthasarathy}},\ }\bibfield  {title} {\enquote {\bibinfo {title} {On
  estimating the state of a finite level quantum system},}\ }\href@noop {}
  {\bibfield  {journal} {\bibinfo  {journal} {Infin. Dimens. Anal. Quantum
  Probab. Relat. Top.}\ }\textbf {\bibinfo {volume} {7}},\ \bibinfo {pages}
  {607--617} (\bibinfo {year} {2004})}\BibitemShut {NoStop}%
\bibitem [{\citenamefont {Pittenger}\ and\ \citenamefont
  {Rubin}(2004)}]{Pittenger:2004vs}%
  \BibitemOpen
  \bibfield  {author} {\bibinfo {author} {\bibfnamefont {A.~O.}\ \bibnamefont
  {Pittenger}}\ and\ \bibinfo {author} {\bibfnamefont {M.~H.}\ \bibnamefont
  {Rubin}},\ }\bibfield  {title} {\enquote {\bibinfo {title} {Mutually unbiased
  bases, generalized spin matrices and separability},}\ }\href@noop {}
  {\bibfield  {journal} {\bibinfo  {journal} {Linear Algebra Appl.}\ }\textbf
  {\bibinfo {volume} {390}},\ \bibinfo {pages} {255--278} (\bibinfo {year}
  {2004})}\BibitemShut {NoStop}%
\bibitem [{\citenamefont {Durt}(2005)}]{Durt:2005yw}%
  \BibitemOpen
  \bibfield  {author} {\bibinfo {author} {\bibfnamefont {T.}~\bibnamefont
  {Durt}},\ }\bibfield  {title} {\enquote {\bibinfo {title} {About mutually
  unbiased bases in even and odd prime power dimensions},}\ }\href@noop {}
  {\bibfield  {journal} {\bibinfo  {journal} {J. Phys. A}\ }\textbf {\bibinfo
  {volume} {38}},\ \bibinfo {pages} {5267--5284} (\bibinfo {year}
  {2005})}\BibitemShut {NoStop}%
\bibitem [{\citenamefont {Klimov}, \citenamefont {S{{\'a}}nchez-Soto},\ and\
  \citenamefont {de~Guise}(2005)}]{Klimov:2005qb}%
  \BibitemOpen
  \bibfield  {author} {\bibinfo {author} {\bibfnamefont {A.~B.}\ \bibnamefont
  {Klimov}}, \bibinfo {author} {\bibfnamefont {L.~L.}\ \bibnamefont
  {S{{\'a}}nchez-Soto}}, \ and\ \bibinfo {author} {\bibfnamefont
  {H.}~\bibnamefont {de~Guise}},\ }\bibfield  {title} {\enquote {\bibinfo
  {title} {Multicomplementary operators via finite {F}ourier transform},}\
  }\href@noop {} {\bibfield  {journal} {\bibinfo  {journal} {J. Phys. A}\
  }\textbf {\bibinfo {volume} {38}},\ \bibinfo {pages} {2747--2760} (\bibinfo
  {year} {2005})}\BibitemShut {NoStop}%
\bibitem [{\citenamefont {Kibler}\ and\ \citenamefont
  {Planat}(2006)}]{Kibler:2006kn}%
  \BibitemOpen
  \bibfield  {author} {\bibinfo {author} {\bibfnamefont {M.~R.}\ \bibnamefont
  {Kibler}}\ and\ \bibinfo {author} {\bibfnamefont {M.}~\bibnamefont
  {Planat}},\ }\bibfield  {title} {\enquote {\bibinfo {title} {A {SU(2)} recipe
  for mutually unbiased bases},}\ }\href@noop {} {\bibfield  {journal}
  {\bibinfo  {journal} {Int. J. Mod. Phys. B}\ }\textbf {\bibinfo {volume}
  {20}},\ \bibinfo {pages} {1802--1807} (\bibinfo {year} {2006})}\BibitemShut
  {NoStop}%
\bibitem [{\citenamefont {Lawrence}, \citenamefont {\v{C}. Brukner},\ and\
  \citenamefont {Zeilinger}(2002)}]{Lawrence:2002ij}%
  \BibitemOpen
  \bibfield  {author} {\bibinfo {author} {\bibfnamefont {J.}~\bibnamefont
  {Lawrence}}, \bibinfo {author} {\bibnamefont {\v{C}. Brukner}}, \ and\
  \bibinfo {author} {\bibfnamefont {A.}~\bibnamefont {Zeilinger}},\ }\bibfield
  {title} {\enquote {\bibinfo {title} {Mutually unbiased binary observable sets
  on $n$ qubits},}\ }\href@noop {} {\bibfield  {journal} {\bibinfo  {journal}
  {Phys. Rev. A}\ }\textbf {\bibinfo {volume} {65}},\ \bibinfo {pages} {032320}
  (\bibinfo {year} {2002})}\BibitemShut {NoStop}%
\bibitem [{\citenamefont {Lawrence}(2004)}]{Lawrence:2004xj}%
  \BibitemOpen
  \bibfield  {author} {\bibinfo {author} {\bibfnamefont {J.}~\bibnamefont
  {Lawrence}},\ }\bibfield  {title} {\enquote {\bibinfo {title} {Mutually
  unbiased bases and trinary operator sets for n qutrits},}\ }\href {\doibase
  10.1103/PhysRevA.70.012302} {\bibfield  {journal} {\bibinfo  {journal} {Phys.
  Rev. A}\ }\textbf {\bibinfo {volume} {70}},\ \bibinfo {pages} {012302}
  (\bibinfo {year} {2004})}\BibitemShut {NoStop}%
\bibitem [{\citenamefont {Romero}\ \emph {et~al.}(2005)\citenamefont {Romero},
  \citenamefont {Bj{\"{o}}rk}, \citenamefont {Klimov},\ and\ \citenamefont
  {S{{\'a}}nchez-Soto}}]{Romero:2005dz}%
  \BibitemOpen
  \bibfield  {author} {\bibinfo {author} {\bibfnamefont {J.~L.}\ \bibnamefont
  {Romero}}, \bibinfo {author} {\bibfnamefont {G.}~\bibnamefont {Bj{\"{o}}rk}},
  \bibinfo {author} {\bibfnamefont {A.~B.}\ \bibnamefont {Klimov}}, \ and\
  \bibinfo {author} {\bibfnamefont {L.~L.}\ \bibnamefont
  {S{{\'a}}nchez-Soto}},\ }\bibfield  {title} {\enquote {\bibinfo {title}
  {Structure of the sets of mutually unbiased bases for $n$ qubits},}\
  }\href@noop {} {\bibfield  {journal} {\bibinfo  {journal} {Phys. Rev. A}\
  }\textbf {\bibinfo {volume} {72}},\ \bibinfo {pages} {062310} (\bibinfo
  {year} {2005})}\BibitemShut {NoStop}%
\bibitem [{\citenamefont {Garcia}, \citenamefont {Romero},\ and\ \citenamefont
  {Klimov}(2010)}]{Garcia:2010hr}%
  \BibitemOpen
  \bibfield  {author} {\bibinfo {author} {\bibfnamefont {A.}~\bibnamefont
  {Garcia}}, \bibinfo {author} {\bibfnamefont {J.~L.}\ \bibnamefont {Romero}},
  \ and\ \bibinfo {author} {\bibfnamefont {A.~B.}\ \bibnamefont {Klimov}},\
  }\bibfield  {title} {\enquote {\bibinfo {title} {Generation of bases with
  definite factorization for an n -qubit system and mutually unbiased sets
  construction},}\ }\href@noop {} {\bibfield  {journal} {\bibinfo  {journal}
  {J. Phys. A}\ }\textbf {\bibinfo {volume} {43}},\ \bibinfo {pages} {385301}
  (\bibinfo {year} {2010})}\BibitemShut {NoStop}%
\bibitem [{\citenamefont {Wie{\'s}niak}, \citenamefont {Paterek},\ and\
  \citenamefont {Zeilinger}(2011)}]{Wiesniak:2011kl}%
  \BibitemOpen
  \bibfield  {author} {\bibinfo {author} {\bibfnamefont {M.}~\bibnamefont
  {Wie{\'s}niak}}, \bibinfo {author} {\bibfnamefont {T.}~\bibnamefont
  {Paterek}}, \ and\ \bibinfo {author} {\bibfnamefont {A.}~\bibnamefont
  {Zeilinger}},\ }\bibfield  {title} {\enquote {\bibinfo {title} {Entanglement
  in mutually unbiased bases},}\ }\href@noop {} {\bibfield  {journal} {\bibinfo
   {journal} {New J. Phys.}\ }\textbf {\bibinfo {volume} {13}},\ \bibinfo
  {pages} {053047} (\bibinfo {year} {2011})}\BibitemShut {NoStop}%
\bibitem [{\citenamefont {{\v R}eh{\'a}{\v c}ek}\ \emph
  {et~al.}(2013)\citenamefont {{\v R}eh{\'a}{\v c}ek}, \citenamefont {Hradil},
  \citenamefont {Klimov}, \citenamefont {Leuchs},\ and\ \citenamefont
  {S{\'a}nchez-Soto}}]{Rehacek:2013jt}%
  \BibitemOpen
  \bibfield  {author} {\bibinfo {author} {\bibfnamefont {J.}~\bibnamefont {{\v
  R}eh{\'a}{\v c}ek}}, \bibinfo {author} {\bibfnamefont {Z.}~\bibnamefont
  {Hradil}}, \bibinfo {author} {\bibfnamefont {A.~B.}\ \bibnamefont {Klimov}},
  \bibinfo {author} {\bibfnamefont {G.}~\bibnamefont {Leuchs}}, \ and\ \bibinfo
  {author} {\bibfnamefont {L.~L.}\ \bibnamefont {S{\'a}nchez-Soto}},\
  }\bibfield  {title} {\enquote {\bibinfo {title} {Sizing up entanglement in
  mutually unbiased bases with fisher information},}\ }\href@noop {} {\bibfield
   {journal} {\bibinfo  {journal} {Phys. Rev. A}\ }\textbf {\bibinfo {volume}
  {88}},\ \bibinfo {pages} {052110} (\bibinfo {year} {2013})}\BibitemShut
  {NoStop}%
\bibitem [{\citenamefont {Spengler}\ and\ \citenamefont
  {Kraus}(2013)}]{Spengler:2013pi}%
  \BibitemOpen
  \bibfield  {author} {\bibinfo {author} {\bibfnamefont {C.}~\bibnamefont
  {Spengler}}\ and\ \bibinfo {author} {\bibfnamefont {B.}~\bibnamefont
  {Kraus}},\ }\bibfield  {title} {\enquote {\bibinfo {title} {Graph-state
  formalism for mutually unbiased bases},}\ }\href@noop {} {\bibfield
  {journal} {\bibinfo  {journal} {Phys. Rev. A}\ }\textbf {\bibinfo {volume}
  {88}},\ \bibinfo {pages} {052323} (\bibinfo {year} {2013})}\BibitemShut
  {NoStop}%
\bibitem [{\citenamefont {Godsil}\ and\ \citenamefont
  {Roy}(2009)}]{Godsil:2009}%
  \BibitemOpen
  \bibfield  {author} {\bibinfo {author} {\bibfnamefont {C.}~\bibnamefont
  {Godsil}}\ and\ \bibinfo {author} {\bibfnamefont {A.}~\bibnamefont {Roy}},\
  }\bibfield  {title} {\enquote {\bibinfo {title} {Equiangular lines, mutually
  unbiased bases, and spin models},}\ }\href@noop {} {\bibfield  {journal}
  {\bibinfo  {journal} {European J. Combin.}\ }\textbf {\bibinfo {volume}
  {30}},\ \bibinfo {pages} {246--262} (\bibinfo {year} {2009})}\BibitemShut
  {NoStop}%
\bibitem [{\citenamefont {Kantor}(2012)}]{Kantor:2012}%
  \BibitemOpen
  \bibfield  {author} {\bibinfo {author} {\bibfnamefont {W.~M.}\ \bibnamefont
  {Kantor}},\ }\bibfield  {title} {\enquote {\bibinfo {title} {Mubs
  inequivalence and affine planes},}\ }\href@noop {} {\bibfield  {journal}
  {\bibinfo  {journal} {J. Math. Phys.}\ }\textbf {\bibinfo {volume} {53}},\
  \bibinfo {pages} {032204} (\bibinfo {year} {2012})}\BibitemShut {NoStop}%
\bibitem [{\citenamefont {Chau}(2005)}]{Chau:2005wj}%
  \BibitemOpen
  \bibfield  {author} {\bibinfo {author} {\bibfnamefont {H.~F.}\ \bibnamefont
  {Chau}},\ }\bibfield  {title} {\enquote {\bibinfo {title} {Unconditionally
  secure key distribution in higher dimensions by depolarization},}\
  }\href@noop {} {\bibfield  {journal} {\bibinfo  {journal} {IEEE Trans. Inf.
  Theory}\ }\textbf {\bibinfo {volume} {51}},\ \bibinfo {pages} {1451--1568}
  (\bibinfo {year} {2005})}\BibitemShut {NoStop}%
\bibitem [{\citenamefont {Gow}(2007)}]{Gow:2007dz}%
  \BibitemOpen
  \bibfield  {author} {\bibinfo {author} {\bibfnamefont {R.}~\bibnamefont
  {Gow}},\ }\bibfield  {title} {\enquote {\bibinfo {title} {Generation of
  mutually unbiased bases as powers of a unitary matrix in 2-power
  dimensions},}\ }\href@noop {} {\bibfield  {journal} {\bibinfo  {journal}
  {arXiv:math/0703333v2}\ } (\bibinfo {year} {2007})}\BibitemShut {NoStop}%
\bibitem [{\citenamefont {Kern}, \citenamefont {Ranade},\ and\ \citenamefont
  {Seyfarth}(2010)}]{Kern:2010lc}%
  \BibitemOpen
  \bibfield  {author} {\bibinfo {author} {\bibfnamefont {O.}~\bibnamefont
  {Kern}}, \bibinfo {author} {\bibfnamefont {K.~S.}\ \bibnamefont {Ranade}}, \
  and\ \bibinfo {author} {\bibfnamefont {U.}~\bibnamefont {Seyfarth}},\
  }\bibfield  {title} {\enquote {\bibinfo {title} {Complete sets of cyclic
  mutually unbiased bases in even prime-power dimensions},}\ }\href@noop {}
  {\bibfield  {journal} {\bibinfo  {journal} {J. Phys. A}\ }\textbf {\bibinfo
  {volume} {43}},\ \bibinfo {pages} {275305} (\bibinfo {year}
  {2010})}\BibitemShut {NoStop}%
\bibitem [{\citenamefont {Seyfarth}\ and\ \citenamefont
  {Ranade}(2011)}]{Seyfarth:2011ru}%
  \BibitemOpen
  \bibfield  {author} {\bibinfo {author} {\bibfnamefont {U.}~\bibnamefont
  {Seyfarth}}\ and\ \bibinfo {author} {\bibfnamefont {K.~S.}\ \bibnamefont
  {Ranade}},\ }\bibfield  {title} {\enquote {\bibinfo {title} {Construction of
  mutually unbiased bases with cyclic symmetry for qubit systems},}\
  }\href@noop {} {\bibfield  {journal} {\bibinfo  {journal} {Phys. Rev. A}\
  }\textbf {\bibinfo {volume} {84}},\ \bibinfo {pages} {042327} (\bibinfo
  {year} {2011})}\BibitemShut {NoStop}%
\bibitem [{\citenamefont {Seyfarth}\ and\ \citenamefont
  {Ranade}(2012)}]{Seyfarth:2012sl}%
  \BibitemOpen
  \bibfield  {author} {\bibinfo {author} {\bibfnamefont {U.}~\bibnamefont
  {Seyfarth}}\ and\ \bibinfo {author} {\bibfnamefont {K.~S.}\ \bibnamefont
  {Ranade}},\ }\bibfield  {title} {\enquote {\bibinfo {title} {Cyclic mutually
  unbiased bases, {F}ibonacci polynomials and {W}iedemann's conjecture},}\
  }\href@noop {} {\bibfield  {journal} {\bibinfo  {journal} {J. Math. Phys.}\
  }\textbf {\bibinfo {volume} {53}},\ \bibinfo {pages} {062201} (\bibinfo
  {year} {2012})}\BibitemShut {NoStop}%
\bibitem [{\citenamefont {Hoggart}\ and\ \citenamefont
  {Bicknell}(1973)}]{Hoggart:1973mq}%
  \BibitemOpen
  \bibfield  {author} {\bibinfo {author} {\bibfnamefont {V.~E.}\ \bibnamefont
  {Hoggart}}\ and\ \bibinfo {author} {\bibfnamefont {M.}~\bibnamefont
  {Bicknell}},\ }\bibfield  {title} {\enquote {\bibinfo {title} {Roots of
  {F}ibonacci polynomials},}\ }\href@noop {} {\bibfield  {journal} {\bibinfo
  {journal} {Fibonacci Quart.}\ }\textbf {\bibinfo {volume} {11}},\ \bibinfo
  {pages} {271--274} (\bibinfo {year} {1973})}\BibitemShut {NoStop}%
\bibitem [{\citenamefont {Goldwasser}, \citenamefont {Klostermeyer},\ and\
  \citenamefont {Ware}(2002)}]{Goldwasser:2002pz}%
  \BibitemOpen
  \bibfield  {author} {\bibinfo {author} {\bibfnamefont {J.~L.}\ \bibnamefont
  {Goldwasser}}, \bibinfo {author} {\bibfnamefont {W.}~\bibnamefont
  {Klostermeyer}}, \ and\ \bibinfo {author} {\bibfnamefont {H.}~\bibnamefont
  {Ware}},\ }\bibfield  {title} {\enquote {\bibinfo {title} {Fibonacci
  polynomials and parity domination in grid graphs.}}\ }\href@noop {}
  {\bibfield  {journal} {\bibinfo  {journal} {Graph. Combinator.}\ }\textbf
  {\bibinfo {volume} {18}},\ \bibinfo {pages} {271--283} (\bibinfo {year}
  {2002})}\BibitemShut {NoStop}%
\bibitem [{\citenamefont {Gow}(1980)}]{Gow:1980hx}%
  \BibitemOpen
  \bibfield  {author} {\bibinfo {author} {\bibfnamefont {R.}~\bibnamefont
  {Gow}},\ }\bibfield  {title} {\enquote {\bibinfo {title} {The equivalence of
  an invertible matrix to its transpose},}\ }\href@noop {} {\bibfield
  {journal} {\bibinfo  {journal} {Linear Multilinear A.}\ }\textbf {\bibinfo
  {volume} {8}},\ \bibinfo {pages} {329--336} (\bibinfo {year}
  {1980})}\BibitemShut {NoStop}%
\bibitem [{\citenamefont {Djokovic}, \citenamefont {Szechtman},\ and\
  \citenamefont {Zhao}(2003)}]{Djokovic:2003qt}%
  \BibitemOpen
  \bibfield  {author} {\bibinfo {author} {\bibfnamefont {D.~Z.}\ \bibnamefont
  {Djokovic}}, \bibinfo {author} {\bibfnamefont {F.}~\bibnamefont {Szechtman}},
  \ and\ \bibinfo {author} {\bibfnamefont {K.}~\bibnamefont {Zhao}},\
  }\bibfield  {title} {\enquote {\bibinfo {title} {An algorithm that carries a
  square matrix into its transpose by an involutory congruence
  transformation},}\ }\href@noop {} {\bibfield  {journal} {\bibinfo  {journal}
  {Electron. J. Linear Al.}\ }\textbf {\bibinfo {volume} {10}},\ \bibinfo
  {pages} {320--340} (\bibinfo {year} {2003})}\BibitemShut {NoStop}%
\bibitem [{\citenamefont {Venkaiah}\ and\ \citenamefont
  {Sen}(1988)}]{Venkaiah:1988ph}%
  \BibitemOpen
  \bibfield  {author} {\bibinfo {author} {\bibfnamefont {V.~C.}\ \bibnamefont
  {Venkaiah}}\ and\ \bibinfo {author} {\bibfnamefont {S.}~\bibnamefont {Sen}},\
  }\bibfield  {title} {\enquote {\bibinfo {title} {Computing a matrix
  symmetrizer exactly using modified multiple modulus residue arithmetic},}\
  }\href@noop {} {\bibfield  {journal} {\bibinfo  {journal} {J. Comp. Appl.
  Math.}\ }\textbf {\bibinfo {volume} {21}},\ \bibinfo {pages} {27--40}
  (\bibinfo {year} {1988})}\BibitemShut {NoStop}%
\end{thebibliography}
\end{document}